\documentclass[letterpaper,twoside,twocolumn,
  superscriptaddress,
  aps,pra,floatfix]{revtex4-1}
\usepackage{mathrsfs}
\usepackage{amssymb}
\usepackage[dvips]{graphicx}
\usepackage{color}
\usepackage[hypertex]{hyperref}

\usepackage{bm}
\usepackage{amsfonts}
\usepackage{amsthm}

\newtheorem{statement}{Lemma}
\newtheorem{theorem}{Theorem}
\newtheorem{corollary}{Corollary}[theorem]
\newtheorem{example}{Example}

\def\span{\mathop{\rm span}}
\def\ord{\mathop{\rm ord}}
\def\wgt{\mathop{\rm wgt}}
\def\mod{\mathop{\rm mod}}

\def\clg{\mathop{\mathrm{Cl}_{\mathcal{G}}}}

\newcommand{\ket}[1]{\left\vert{#1}\right\rangle}

\makeatletter
\def\Hy@safe@activestrue{}
\makeatother
\def\urlprefix#1#2{\hskip0pt plus0.01fil\discretionary{}{}{}%
  \hbox{\url{#2}}}

\begin{document}

\author{Alexey A. Kovalev}
\affiliation{Department of Physics \& Astronomy, University of
  California, Riverside, California 92521, USA}
\author{Ilya Dumer}
\affiliation{Department of Electrical Engineering, University of
  California, Riverside, California 92521, USA}
\author{Leonid P. Pryadko}
\affiliation{Department of Physics \& Astronomy, University of
  California, Riverside, California 92521, USA}

\title{Low-complexity quantum codes designed via codeword-stabilized framework}
\date\today
\begin{abstract}
  We consider design of the quantum stabilizer codes via a two-step,
  low-complexity approach based on the framework of
  codeword-stabilized (CWS) codes. In this framework, each quantum CWS
  code can be specified by a graph and a binary code.  For codes that
  can be obtained from a given graph, we give several upper bounds on
  the distance of a generic (additive or non-additive) CWS code, and
  the lower Gilbert-Varshamov bound for the existence of additive CWS
  codes.  We also consider additive cyclic CWS codes and show that
  these codes correspond to a previously unexplored class of
  single-generator cyclic stabilizer codes.  We present several families
  of simple stabilizer codes with relatively good parameters.

\end{abstract}
\maketitle

\section{Introduction}

It was the invention of quantum error correcting codes
\cite{shor-error-correct,Knill-Laflamme-1997,Bennett-1996} (QECCs)
that opened quantum computing (QC) as a theoretical possibility.
However, high precision required for error correction
\cite{Knill-error-bound,Rahn-2002,Dennis-Kitaev-Landahl-Preskill-2002,%
  Steane-2003,%
  Fowler-QEC-2004,Fowler-2005,fowler-thesis-2005,%
  Knill-nature-2005,Knill-2005,Raussendorf-Harrington-2007} combined
with the large number of auxiliary qubits necessary to implement
it, have so far inhibited any practical realization beyond
proof-of-the-principle
demonstrations\cite{chuang-2000,chuang-2001,Gulde-2003,%
  Chiaverini-2004,Friedenauer-2008,Martinis-review-2009,%
  Kim-2010}. 

In any QECC, one needs to perform certain many-body quantum
measurements in order to decide how to correct the encoded state.
The practical difficulty is that  a generic code requires
measurements which are both complicated and frequent at the same
time. It is therefore clear that a quantum computer can only be
build via a thorough optimization at every step of the design. In
particular, code optimization targets codes that combine good
parameters with fairly simple measurements.  It is also desirable
to parallelize these measurements  given a specific on-chip layout
of a QC architecture.

To date, the main focus of the QECC-research has been on finding good codes
with the traditional code parameters, which are the block length $n$, code
dimension $K$, and code distance $d$ (or code rate $R\equiv (\log_2K)/n$ and
the relative distance $\delta\equiv d/n$).  For stabilizer codes
\cite{gottesman-thesis,Calderbank-1997}, we also consider the number of
encoded qubits $k=\log_2K$.

A number of stabilizer codes\cite{Grassl:codetables} have been
designed that meet or nearly achieve the existing bounds on
distance $d$ for the given $k$ and $n$. Code parameters can be
further refined by going beyond the family of stabilizer codes.
One example is a recently introduced framework of
codeword-stabilized (CWS) quantum
codes\cite{Smolin-2007,Cross-CWS-2009,Chuang-CWS-2009,Chen-Zeng-Chuang-2008}.
 A qubit CWS code ${\cal Q}\equiv ({\cal G},{\cal C})$ (in standard form) is
determined by a graph ${\cal G}$ and a classical binary code ${\cal
  C}$.  CWS codes include all stabilizer codes as a subclass (the
corresponding binary code ${\cal C}$ must be linear), but also the
codes which have been proved to have parameters superior to those of
any stabilizer code\cite{Cross-CWS-2009,Yu-2008,Yu-Chen-Oh-2007,%
  Grassl-Roetteler-2008A,Grassl-Roetteler-2008B,Grassl-2009}.
Unfortunately, typical gains in code dimension $K$ correspond to a
fraction of a qubit.  Moreover, error-correcting algorithms known
for general non-additive CWS codes have exponential
complexity\cite{Li-Dumer-Pryadko-2010,Li-Dumer-Grassl-Pryadko-2010},
as opposed to polynomial complexity of the stabilizer codes.

Even for the relatively simple additive codes, their optimization is a very
difficult problem that has exponential complexity. This is one of the main
reasons as to why the two relatively simple code families are almost
exclusively used among stabilizer codes to estimate the threshold accuracy
required for scalable quantum computation: the concatenated
codes\cite{Knill-error-bound,%
  Rahn-2002,Steane-2003,%
  Fowler-QEC-2004,Fowler-2005,fowler-thesis-2005,%
  Knill-nature-2005,Knill-2005} and the surface
codes\cite{Dennis-Kitaev-Landahl-Preskill-2002,Raussendorf-Harrington-2007}
which originated from the toric codes\cite{kitaev-anyons}.  Both families have
very low code rates that scale as inverse powers of code distance.

In this work we explore how the framework of CWS codes can be used 
to relegate the design of quantum stabilizer codes to classical
binary linear codes in order to simplify the overall design. In
particular, we formulate several theorems framing the parameters
of an additive CWS code which can be obtained from a given graph.
We also suggest a simple decomposition of the $\mathbb{F}_{4}$
generator matrix corresponding to the stabilizer in terms of  the
graph adjacency matrix and the parity check matrix of the binary
code. Finally, we design several graph families corresponding to
regular lattices which result in some particularly good codes.
These include graphs with circulant adjacency matrix which can be
used to construct single-generator cyclic additive codes, a class of
codes overlooked in previous publications. In particular, we prove
the existence of single-generator cyclic additive codes with the
parameters $[[km,k,m]]$, $k>10$ and $[[t^2+(t+1)^2,1,2t+1]]$
(version of toric codes).
Note that these code families have distances that are not bounded, unlike any
CWS code families constructed previously\cite{Cross-CWS-2009,Looi-2008}.

The paper is organized as follows. In Sec.~\ref{sec:notations}, we
introduce the notations and briefly review some known results for
quantum and classical codes. In Sec.~\ref{sec:upper-bound}, we
establish several upper bounds on general CWS codes.  In
Sec.~\ref{sec:additive-CWS} we give a CWS decomposition of the
$\mathbb{F}_4$ matrix corresponding to the stabilizer generators.  In
Sec.~\ref{sec:gv}, we formulate the Gilbert-Varshamov (GV) bounds for
additive CWS codes which can be obtained from a given graph.  Cyclic
additive CWS and more general single-generator additive cyclic codes are
considered in Sec.~\ref{sec:cyclic-CWS} where we discuss their
properties and give several examples. We give our conclusions in
Sec.~\ref{sec:conclusions}.

\section{Notations and some known results}
\label{sec:notations}

\subsection{Classical and quantum error correcting codes}
\label{sec-intro-general-qecc} A \emph{classical} $q$-ary block
error-correcting code $(n,K,d)_q$ is a set of $K$ length-$n$
strings over an alphabet with $q$ symbols. Different strings
represent $K$ distinct messages which can be transmitted.  The
(Hamming) distance between two strings is the number of positions
where they differ.  Distance $d$ of the code ${\cal C}$ is the
minimum distance between any two different strings from ${\cal
C}$.

In the case of \emph{linear} codes, the elements of the alphabet
must form a Galois field $\mathbb{F}_q$; all strings form
$n$-dimensional vector space $\mathbb{F}_q^n$.  A linear
error-correcting code $[n,k,d]_q$ is a $k$-dimensional subspace of
$\mathbb{F}_q^n$.  The distance of a linear code is just the
minimum weight of a non-zero vector in the code, where weight
$\wgt(\mathbf{c})$ of a vector $\mathbf{c}$ is the number of
non-zero elements.  A basis of the code is formed by the rows of its
\emph{generator matrix} $G$.  All vectors that are orthogonal to
the code form the corresponding $(n-k)$-dimensional dual code, its
generator matrix is the parity-check matrix $H$ of the original
code.

For a \emph{binary} code ${\cal C}[n,k,d]$, the field is just
$\mathbb{F}_2=\{0,1\}$.  For a \emph{quaternary} code $C$, the
field is $\mathbb{F}_{4}=\{0,1,\omega,\overline{\omega}\}$, with
\begin{equation}
  \label{eq:F4def}
\omega^{2}=\omega+1,\quad \omega^{3}=1,\;\, \mathrm{and}\;\,
\overline{\omega}\equiv \omega^{2}.
\end{equation}

For non-binary codes, there is also a distinct class of
\emph{additive} classical codes, defined as subsets of
$\mathbb{F}_q^n$ closed under addition (in the binary case these
are just linear codes).  A code ${\cal C}$ is cyclic if inclusion
$(c_{0},c_{1},\ldots ,c_{n-1})\in{\cal C}$ implies that
$(c_{n-1},c_{0},c_{1},\ldots ,c_{n-2})\in{\cal C}$.  Codes that
are both linear and cyclic are particularly simple: by mapping
vectors to polynomials in the natural way, $\mathbf{c}\to
c(x)\equiv c_0+c_1 x+\ldots+c_{n-1}x^{n-1}$, it is possible to
show that any such code consists of polynomials which are
multiples of a single generator polynomial $g(x)$, which must
divide $x^{n}-1$ (using the algebra corresponding to the field
$\mathbb{F}_q$).  The quotient defines the \emph{check polynomial}
$h(x)$,
\begin{equation}
  h(x) g(x)=x^n-1\label{eq:linear-cyclic-condition}
\end{equation}
which  is the generator polynomial of the dual code.  The degree
of the generator polynomial is $\deg g(x)=n-k$.  The corresponding
generator matrix $G$ can be chosen as (the first $k$ rows of) the
circulant matrix formed by subsequent shifts of the vector that
corresponds to $g(x)$.

\emph{Qubit} quantum error correcting codes are defined in the complex
Hilbert space $\mathcal{H}_{2}^{\otimes n}$, where $\mathcal{H}_{2}$
is the Hilbert space of a single two-level system.  $\mathcal{H}_{2}$
is formed by all vectors $\alpha\left|0\right\rangle
+\beta\left|1\right\rangle $ with $\alpha,\beta\in\mathbb{C}$, and the
inner product such that the two states are orthonormal, $\langle
  i| j\rangle=\delta_{ij}$, $i,j\in \{0,1\}$.
Any operator
acting in $\mathcal{H}_{2}^{\otimes n}$ can be represented as a linear
combination of Pauli operators which form the $n$-qubit Pauli group
$\mathscr{P}_{n}$ of size $2^{2n+2}$,
\begin{equation}
  \mathscr{P}_{n}=i^{m}\{I,X,Y,Z\}^{\otimes n},\; m=0,\ldots ,3\:,
  \label{PauliGroup}
\end{equation}
where $X$, $Y$, and $Z$ are the usual Pauli matrices, and $I$ is the
identity matrix.  The weight $\wgt(E)$ of a Pauli operator $E$ is the
number of non-identity terms in the corresponding tensor product.

All Pauli operators are unitary; they are also Hermitian with
eigenvalues $\pm 1$ when the phase factor $i^m$ in
Eq.~(\ref{PauliGroup}) is real-valued, $m=0,2$.  A state
$\left|\psi\right\rangle $ is stabilized by a Hermitian Pauli operator $M$ if
$M\left|\psi\right\rangle =\left|\psi\right\rangle $. A linear space
$Q$ is stabilized by a set of operators $\mathcal{M}$ if each vector
in $Q$ is stabilized by every operator in $\mathcal{M}$.

An $((n,K,d))$ \emph{quantum error-correcting code} is a
$K$-dimensional subspace of the Hilbert space
$\mathcal{H}_{2}^{\otimes n}$. Such a subspace can be described by an
orthonormal basis $\{\left|i\right\rangle \}_{i=1}^{K}$. Let
$\mathcal{E}\subset\mathscr{P}_{n}$ be some set of Pauli errors.  A
QECC detects all errors $E\in\mathcal{E}$ if and only if
\cite{Nielsen-book,gottesman-thesis}
\begin{equation}
\left\langle j\left|E\right|i\right\rangle =C_{E}\delta_{ij}\:,
\end{equation}
where $C_{E}$ only depends on the error $E$, but is independent of the
basis vectors.  A QECC has distance $d$ if it can detect all Pauli
errors of weight $(d-1)$, but not all errors of weight $d$.  The
errors in the set $\mathcal{E}$ can be corrected if and only if all
the nontrivial pairwise combinations of errors from $\mathcal{E}$ are
detectable \cite{Bennett-1996,Knill-Laflamme-1997}. Thus a
distance-$d$ code corrects all errors of weight $s\leq
t\equiv\left\lfloor (d-1)/2\right\rfloor $.

The code ${\cal Q}$ is \emph{non-degenerate} if linearly-independent
errors from ${\cal E}$ produce corrupted spaces $E({\cal Q})\equiv
\{E\ket \psi:\ket\psi\in{\cal Q}\}$ that are linearly independent.
Otherwise, the code is \emph{degenerate}, implying the existence of at
least two \emph{mutually degenerate} linearly independent operators
$\{E_1, E_2\}\in {\cal E}$ which act identically on ${\cal Q}$.

The code is called \emph{pure} if linearly independent errors from
${\cal E}$ produce corrupted spaces that are not only linearly
independent, but also mutually orthogonal.  For all codes considered
in this work, non-degenerate codes are also
pure\cite{Calderbank-1997,Li-Dumer-Grassl-Pryadko-2010}.

Two codes are considered equivalent if they differ just by qubit
order, and/or discrete rotations leaving each of the single-qubit
Pauli groups invariant.  The latter are called local Clifford (LC)
transformations.

\subsection{Stabilizer quantum error correcting codes}
\label{sec:intro-stab-codes}

Here we briefly review the well-known family of stabilizer codes
\cite{gottesman-thesis}.  An $[[n,k,d]]$ \emph{stabilizer code} $Q$ is
a $2^{k}$-dimensional subspace of the Hilbert space
$\mathcal{H}_{2}^{\otimes n}$ stabilized by an Abelian group
$\mathscr{S}\subset\mathscr{P}_{n}$ with $n-k$ Hermitian Pauli generators,
$\mathscr{S}=\left\langle G_{1},\ldots ,G_{n-k}\right\rangle $.
Explicitly,
\begin{equation}\label{eq:stabilizer}
  Q\equiv\{\left|\psi\right\rangle :
  S\left|\psi\right\rangle =\left|\psi\right\rangle ,
  \forall S\in\mathscr{S}\}\:.
\end{equation}
Such a code exists only if $-\openone\notin\mathscr{S}$.  The group
$\mathscr{S}$ is called the stabilizer of the code.  Changing the
sign(s) of one or several of the generators $G_i$ results in replacing
$Q$ with one of $2^{n-k}-1$ equivalent codes whose direct sum (together
with $Q$) is the entire space $\mathscr{P}_n$.

The \emph{normalizer} of $\mathscr{S}$ is a set of Pauli
operators generating unitary transformations that leave $\mathscr{S}$
invariant,
\begin{equation}
  \mathscr{N}\equiv\{U\in\mathscr{P}_{n}:
  U^{\dagger}SU=S,\forall S\in\mathscr{S}\}\:.
\end{equation}
Elements of the normalizer form a group commuting with $\mathscr{S}$
but not necessarily with each other.  It is possible to construct $2k$
logical operators $\overline{X}_{j}$, $\overline{Z}_{j}$, $j=1,\ldots
,k$ belonging to $\mathscr{P}_{n}$ with the usual commutation
relations that generate the normalizer when the generators of
$\mathscr{S}$ are included \cite{gottesman-thesis,Wilde-2009}.  The
Abelian subgroup of $\mathscr{N}$, $\mathscr{S}_{0}=\left\langle
  G_{1},\ldots ,G_{n-k},\overline{Z}_{1},\ldots
  ,\overline{Z}_{k}\right\rangle $, becomes a maximal Abelian subgroup
of $\mathscr{P}_{n}$ when the generator $i\openone$ is also
included.

The group $\mathscr{S}_{0}$ stabilizes a unique \emph{stabilizer}
state $\left|s\right\rangle \equiv\overline{\left|0\ldots
    0\right\rangle }$, an $[[n,0,d']]$ stabilizer code, while the
operators $\overline{X}_{j}$ generate the basis of the code, i.e.,
\begin{equation}
  \overline{\left|c_{1}\ldots c_{k}\right\rangle}=\overline{X}_{1}^{c_{1}}\ldots
  \overline{X}_{k}^{c_{k}}\left|s\right\rangle .
  \label{eq:code-basis}
\end{equation}
By convention, the
stabilizer state is considered non-degenerate, and its distance $d'$
is the minimum weight of a non-trivial member of the group $\mathscr{S}_0$.

For stabilizer codes, phases of (Hermitian) Pauli operators are only
needed to choose one of the equivalent codes in
Eq.~(\ref{eq:stabilizer}), as well as to introduce the commutation
relations.  It is convenient to drop the phases and map the
Pauli operators to two binary strings,
$\mathbf{v},\mathbf{u}\in\{0,1\}^{n}$ \cite{Calderbank-1997},
\begin{equation}
  U\equiv i^{m'}X^{\mathbf{v}} Z^{\mathbf{u}}\rightarrow(\mathbf{v},\mathbf{u}),
  \label{eq:Umap}
\end{equation}
where
$X^{\mathbf{v}}=X_{1}^{v_{1}}X_{2}^{v_{2}}\ldots X_{n}^{v_{n}}$,
$Z^{\mathbf{u}}=Z_{1}^{u_{1}}Z_{2}^{u_{2}}\ldots Z_{n}^{u_{n}}$, and
$m'=0,\ldots,3$ is generally different from that in
Eq.~(\ref{PauliGroup}).  This map preserves the operator algebra, with
a product of two Pauli operators $U_1$ and $U_2$ corresponding to a
sum of the corresponding binary vectors $(\mathbf{v}_1,\mathbf{u}_1)$ and
$(\mathbf{v}_2,\mathbf{u}_2)$.

The map~(\ref{eq:Umap}) can be taken one step
further\cite{Calderbank-1997} to quaternary codes, by introducing
$\mathbb{F}_4^n$ vectors $\mathbf{e}\equiv \mathbf{u}+\omega
\mathbf{v}$ [see Eq.~(\ref{eq:F4def}); note that this mapping differs
  slightly from that in Ref.~\onlinecite{Calderbank-1997}].  We will
denote this combined map as a function $\mathbf{e}\equiv \phi(U)$.
Note that up to a phase this association also allows us to define
$\phi^{-1}(\mathbf{e})$.  To be specific, for the Pauli operator
$\phi^{-1}(\mathbf{e})$ we will set $m'=\mathbf{v}\cdot \mathbf{u}$ in
Eq.~(\ref{eq:Umap}), which corresponds to $m=0$ in
Eq.~(\ref{PauliGroup}).

It is easy to check that two Pauli operators commute if and only if
the symplectic scalar product $\mathbf{v}_1\cdot
\mathbf{u}_2+\mathbf{u}_1\cdot \mathbf{v}_2$ vanishes ($\mod 2$).  In
terms of the corresponding $\{\mathbf{e}_1, \mathbf{e}_2\}\subset \mathbb{F}_4$,
this corresponds to the vanishing of the \emph{trace inner product}
\begin{equation}
  \mathbf{e}_{1}*\mathbf{e}_{2}\equiv \mathbf{e}_1\cdot \overline \mathbf{e}_2
  +\overline \mathbf{e}_1\cdot \mathbf{e}_2,\label{InnerProduct}
\end{equation}
where $\overline\mathbf{e}_i\equiv
\mathbf{u}_i+\overline\omega\mathbf{v}_i$, $i=0,1$.

A dual code to an additive $\mathbb{F}_{4}$ code $C$ (equipped with
trace inner product) is defined as \cite{Calderbank-1997}
\begin{equation}
  \label{eq:C-perp}
C_{\perp}=\{\mathbf{e}'\in\mathbb{F}_{4}^{n}:\mathbf{e}'*\mathbf{e}=0,\:
\mbox{for}\:\mbox{all}\:\mathbf{e}\in C\}\,.
\end{equation}
If $C\subseteq C_{\perp}$, one says $C$ is self-orthogonal. A
classical additive code $C$ corresponding to a set of operators
$\mathscr{S}_{1}$ is self-orthogonal if and only if $\mathscr{S}_{1}$
is an Abelian group.  Thus any quantum stabilizer code can be
described as a self-orthogonal classical additive code over
$\mathbb{F}_{4}$.  The following theorem is applicable to additive
$\mathbb{F}_{4}$ codes (variant of Theorem 2 from
Ref.~\cite{Calderbank-1997}):

\begin{theorem} \label{Th1}Suppose $C$ is an additive self-orthogonal
code in $\mathbb{F}_{4}^{n}$, containing $2^{n-k}$ vectors, such
that there are no vectors of weight $<d$ in $C_{\perp}\setminus C$.
Then $\phi^{-1}(C)$ defines a stabilizer of an additive QECC with
parameters $[[n,k,d]]$. \end{theorem}

\begin{example}\label{Ex:713} The well-known Calderbank-Shor-Steane (CSS)
  $[[7,1,3]]$ code \cite{Calderbank-Shor-1996,Steane-1996} has the
  stabilizer with the generators\cite{gottesman-thesis}
  \begin{eqnarray}
    \label{eq:CSS-Hamming}
    XXXXIII,\;&XXIIXXI,&XIXIXIX,\nonumber\\
    ZZZZIII,\quad&ZZIIZZI,&ZIZIZIZ,
  \end{eqnarray}
and the logical operators
\begin{equation}
\overline{X}=ZZZZZZZ,\:\overline{Z}=XXXXXXX.
\end{equation}
As any CSS code, this code is linear.  Qubit permutations also give an
equivalent cyclic linear code with the generator polynomial
$g(x)=1+x+x^2+x^4$; $g(x)$ is a factor of $x^7-1$.  The corresponding
check polynomial is $h(x)=1+x+x^3$.
 \end{example}

\subsection{Codeword stabilized codes}

CWS codes \cite{Cross-CWS-2009} represent a general class of
nonadditive QECCs.  A general CWS code is defined in terms of a
stabilizer state $\ket s$ and a set of $K$ mutually commuting
\emph{codeword operators} ${\cal W}\equiv
\{W_i\}_{i=1}^K\subset\mathscr{P}_n$.  Explicitly
[cf.~Eq.~(\ref{eq:code-basis})],
\begin{equation}
  \label{eq:CWS-general}
  {\cal Q}=\span(\{W_i\ket{s}\}_{i=1}^K).
\end{equation}
For non-trivial CWS codes, this construction coincides with
union-stabilizer (USt) codes\cite{Grassl-1997}, restricted to the
zero-dimensional originating code.

Any stabilizer state is LC-equivalent 
to a \emph{graph state} \cite{Grassl-Klappenecker-Roetteler-2002,%
  Schlingemann-2002,VandenNest-2004,Hein-2006}, a stabilizer state
with the stabilizer group $\mathscr{S}_{{\cal G}}\equiv\langle
S_{1},\ldots,S_{n}\rangle$ whose generators $S_i$ are determined by
the adjacency matrix $R\in\{0,1\}^{n\times n}$ of a (simple) graph
${\cal G}$,
\begin{equation}
  \label{eq:graph-generators}
S_{i}=X_{i}Z^{\mathbf{r}_{i}}\,,
\end{equation}
where $\mathbf{r}_{i}$, $i=1,\ldots ,n$ denotes the $i$-th row of $R$.
In fact, such a graph is usually not unique, even after accounting for
graph isomorphisms.  The full set of LC-equivalent graph states can be
generated by a sequence of \emph{local complementations}, operations
on a graph where the subgraph corresponding to a neighborhood of a
particular vertex is inverted.  Such graphs are called \emph{locally
  equivalent}\cite{Bouchet-1993}.

Any CWS code $((n,K,d))$ is LC equivalent to a CWS code in
\emph{standard form}, defined by an order-$n$ graph ${\cal G}$ and a
classical binary code ${\cal C}$ containing $K$ binary words.  The
graph defines the graph state, while the vectors of the classical code
$\mathbf{c}_i\in {\cal C}$ are used to generate the code word
operators, $W_i=Z^{\mathbf{c}_i}$.  Thus,
\begin{equation}
  \label{eq:CWS-basis}
  {\cal Q}=\span(\{ Z^{\mathbf{c}_i}\ket s\}_{ i=1}^K).
\end{equation}
It is customary to use notation $Q=({\cal G},{\cal C})$ for CWS codes
in standard form.

The key simplification of the CWS construction comes from the fact
that the basis states $W_i\ket s$ are eigenvectors of the graph
stabilizer generators,
\begin{equation}
\label{eq:eigen}
S_i W_i\ket s =\pm W_i \ket s,\quad S_i\in\mathscr{S}_{\cal G}.
\end{equation}
Thus, a Pauli operator in the form~(\ref{eq:Umap}) can be transformed
to a $Z$-only operator $Z^{\clg(U)}$, where the \emph{graph image} of
the operator $U$ is the binary vector
\begin{equation}
  \mbox{Cl}_{{\cal G}}(U)\equiv \mathbf{u}+\sum_{i=1}^{n}\mathbf{v}_{i}
  \mathbf{r}_{i} \,(\mod 2).\label{ClMapping}
\end{equation}

The error correcting properties of a
quantum CWS code $Q=({\cal G},{\cal C})$ and the classical code ${\cal
  C}$ are related by the following

\begin{theorem}\label{Th2} (after Theorem 3 from Ref.~\cite{Cross-CWS-2009})
  Consider a CWS code ${\cal Q}=({\cal G},{\cal C})$.  An error $E$
  such that $\clg(E)\neq\mathbf{0}$, is detectable in ${\cal Q}$ if
  and only if the binary vector $\clg(E)$ is detectable within the
  code ${\cal C}$. An error $E$ such that $\clg(E)=\mathbf{0}$ is
  detectable in ${\cal Q}$ if and only if
  $Z^{\mathbf{c}}E=EZ^{\mathbf{c}}$ for all $\mathbf{c}\in{\cal
    C}$. \end{theorem}

The case $\clg(E)\neq \mathbf{0}$ corresponds to pure (non-degenerate)
errors, while $\clg(E)= \mathbf{0}$ indicates that the error is in the
graph stabilizer group $\mathscr{S}_{\cal G}$; the corresponding
detectability condition is a requirement that the error must be
degenerate.

While in general CWS codes are non-additive, they include all
stabilizer codes as a subclass.  A CWS code ${\cal Q}=({\cal G},{\cal
  C})$ is additive if ${\cal C}$ is a linear
code\cite{Cross-CWS-2009}.  The stabilizer $\mathscr{S}$ of an
additive CWS code in standard form is a subgroup of the graph
stabilizer $\mathscr{S}_{\cal G}$; it can be obtained from the
graph-stabilizer generators by a symplectic Gram-Schmidt
orthogonalization procedure\cite{Li-Dumer-Grassl-Pryadko-2010}.
Conversely, the representation~(\ref{eq:code-basis}) of an additive
code corresponds to a general CWS code; an LC transformation may be
needed to obtain the corresponding standard form, and one can always
find a standard form where ${\cal C}$ is
linear\cite{Chuang-CWS-2009}.  In the following we will always assume
such a representation.

\begin{figure}[htbp]
  \centering
  \includegraphics[scale=0.85]{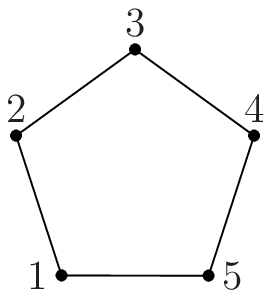}\hfill
  \includegraphics[scale=0.85]{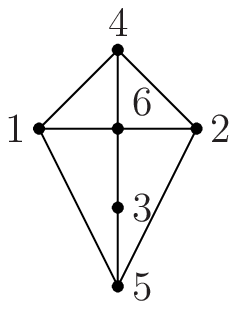}\hfill
  \includegraphics[scale=0.85]{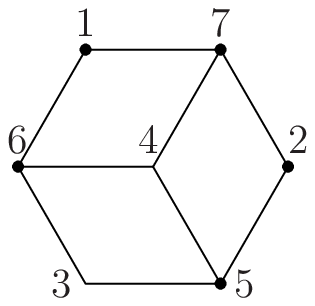}
  \caption{Left: $5$-ring graph corresponding to the $[[5,1,3]]$ code in
    Example Ex.~\ref{Ex:513}.  Center: ``Kite'' graph corresponding to
    the degenerate code
    $[[6,1,3]]$ from Example \ref{Ex:613}.  Right: The graph
    corresponding to the cyclic $[[7,1,3]]$ code from
    Example~\ref{Ex:713cws}.}
  \label{fig:ex613}
\end{figure}

\begin{example}
  \label{Ex:513}
  The smallest single-error-correcting code is the linear cyclic
  $[[5,1,3]]$ code\cite{Bennett-1996,%
    Calderbank-Rains-Shor-Sloane-1997,Laflamme-1996} with the
  generator polynomial $g(x)=1+\omega x+\omega x^2+x^3$ which divides
  $x^5-1$.  This code is unique; its stabilizer generators can be
  obtained as cyclic permutations of a single operator, $XZZXI$, and
  the logical operators are
  \begin{equation}
    \overline{X}=ZZZZZ,\:\overline{Z}=XXXXX.
  \end{equation}
  The corresponding CWS code\cite{Cross-CWS-2009} can be generated
  from the $5$-ring graph in Fig.~\ref{fig:ex613} (left), and the
  binary code has a single generator $\mathbf{c}=(11111)$.  Note that
  both the graph and the binary code preserve the original cyclic
  symmetry.
\end{example}

\begin{example}\label{Ex:613} There exist only two inequivalent 
  single-correcting codes $[[6,1,3]]$; both are
  degenerate\cite{Calderbank-Shor-1996}.  One of the codes is obtained
  from the code in Example \ref{Ex:513} by adding a qubit; the graph
  of the corresponding CWS code can be chosen as a $5$-ring
  [Fig.~\ref{fig:ex613} left] and a disconnected vertex $i=6$; the
  binary code ${\cal C}$ is generated by $\mathbf{c}=(111110)$.  The
  degeneracy group is generated by $S_6=X_6$.  The stabilizer
  generators for the second code are listed in
  Ref.~\cite{Calderbank-Shor-1996}.  This code corresponds to the
  graph in Fig.~\ref{fig:ex613} (center), while the binary code is
  generated by $c_1=(011100)$.  While there are three bits which are
  not involved with the classical code, they cannot be dropped as they
  are part of the entangled state.  The degeneracy group is generated
  by $S_1S_2\equiv X_1X_2$ (the equivalence follows from the fact that
  the first two vertices of ${\cal G}$ share all of their neighbors).
\end{example}

\begin{example}\label{Ex:713cws} The linear cyclic
  $[[7,1,3]]$ CSS code from the Example \ref{Ex:713} is LC equivalent to a CWS
  code with the graph in Fig.~\ref{fig:ex613} (Right). The corresponding
  classical code is given by ${\cal C}=\{0000000,1110000\}$.  Note that
  neither the graph nor the binary code is explicitly symmetric with respect
  to cyclic permutations of the qubits.  Note also that an inequivalent CWS
  cyclic $[[7,1,3]]$ code exists; we constructed such a code among others in
  Example~\ref{ex:k-copies}, see Table~\ref{tab:quantum-repetition}.
 \end{example}

\section{Upper bounds for CWS codes}
\label{sec:upper-bound}

In this section we give upper bounds on general CWS codes in terms of
the properties of the corresponding graph ${\cal G}$ and the binary
code ${\cal C}$.

\begin{statement}
  \label{lemma:upper-binary}
  The distance $d$ of the CWS code ${\cal Q}=({\cal
    G},{\cal C})$ cannot exceed that of $\mathcal{C}$.
\end{statement} \begin{proof} Indeed, any ``classical'' error in the
  form $E=i^mZ^\mathbf{u}$ is mapped by Eq.~(\ref{ClMapping}) to the
  binary vector $\mathbf{u}$.  If $E$ is detectable by ${\cal Q}$,
  $\mathbf{u}$ should be detectable by ${\cal C}$.
\end{proof}

Lemma \ref{lemma:upper-binary} concerns with errors which are dealt
with by the binary code.  On the other hand, a CWS code is an
enlargement of the code formed by the graph state.  The following
observation has been made in Ref.~\cite{Grassl-2009}:
\begin{statement}
  \label{lemma:upper-pure}
  The distance $d$ of a nondegenerate CWS code $((n,K,d))$ is limited
  by the distance $d'({\cal G})$ of the graph stabilizer state, $d\le
  d'({\cal G})$.
\end{statement}
It follows from the fact that any member of the graph stabilizer is
either a degenerate error, or it is a non-detectable error.  Note,
however, as illustrated by the Example \ref{Ex:613}, in general, the
distance of a CWS code can actually be bigger than that of the graph
stabilizer state.

For a binary code ${\cal C}$, we will say that the $j$-th bit is
\emph{involved} in the code if there are vectors in the code for which
the value of $j$-th bit differ, $c_1^j\neq c_2^j$.  Alternatively, if
the all-zero vector $\mathbf{0}$ is in the code (which can always be
arranged), the condition is that there is a vector $\mathbf{c}\in{\cal
  C}$ where $j$-th bit is non-zero, $c^j\neq0$.

\begin{statement}
  \label{lemma:j-th}
  For a CWS code ${\cal Q}=({\cal G},{\cal C})$ with
  $K>1$, let us assume that $j$-th bit is involved in the code ${\cal
    C}$.  Then the graph-stabilizer generator $S_j$ violates the error detection
  condition in Theorem \ref{Th2}.  \end{statement}
\begin{proof}
  Since the generator $S_j$ is in the graph stabilizer,
  $\clg(S_j)=\mathbf{0}$, one has to check the degenerate condition in
  Theorem \ref{Th2}.  The commutativity of $S_j$ with a given
  $Z^\mathbf{c}$ is determined by the $j$-th bit of $\mathbf{c}$;
  conditions of the Lemma ensure that only one of the two vectors
  commute with $S_j$.
\end{proof}

Note that this means that the code distance cannot exceed that of any
$S_j$ corresponding to bits involved in the binary code.  Since at least
$d$ bits must be involved in the binary code, Lemma \ref{lemma:j-th}
guarantees the following bound
\begin{theorem}\label{Th:max-weight} The distance $d$ of a CWS code
  ${\cal Q}=({\cal G},{\cal C})$ cannot exceed the $d$-th
  largest weight of $S_{i}$, minimized over all graphs that are
  locally-equivalent to ${\cal G}$.
\end{theorem}

We will also be using the following
\begin{corollary}
  For a graph ${\cal G}$ with all vertices of the same degree $r$,
  the distance of a CWS code ${\cal Q}=({\cal G},{\cal C})$ cannot
  exceed $r+1$.
\end{corollary}
In particular, for any ring graph $r=2$, which gives $d\le3$, for any
double-ring graph\cite{Chuang-CWS-2009} $d\le4$, and for a large
enough square lattice wrapped into a torus, $d\le5$.

Obviously, to maximize the distance of a CWS code, one may want to
maximize the distance of the binary code ${\cal C}$.  To this end,
it is a good idea to make sure that every bit is involved in
${\cal C}$. For such codes, we have
\begin{theorem}\label{Th:all-bit} The distance of a CWS code ${\cal Q}=({\cal
    G},{\cal C})$ where the binary code ${\cal C}$ involves all bits
  cannot exceed the \emph{minimum} weight of $S_{i}$,
  minimized over all graphs that are locally-equivalent to ${\cal G}$.
\end{theorem}

\section{Additive CWS codes and quaternary codes}
\label{sec:additive-CWS}

The stabilizer of an additive CWS code ${\cal Q}=({\cal G},{\cal C})$
is a subgroup of the Abelian graph stabilizer $\mathscr{S}_{\cal G}$,
and its generators $G_i\in \mathscr{S}$ can be expressed as products
of graph stabilizer generators $S_i$
\cite{Cross-CWS-2009,Li-Dumer-Grassl-Pryadko-2010}.  Explicitly,
\begin{equation}
G_i=\prod_{j=1}^nS_j^{P_{ij}},\label{eq:stabilizer-operator-CWS}
\end{equation}
where $P\in\{0,1\}^{n-k\times n}$ is the corresponding matrix of
binary coefficients.  With the help of
Eq.~(\ref{eq:graph-generators}), we obtain the following decomposition
for the generator matrix $G$ of the associated additive $\mathbb{F}_4$
code $C$,
\begin{equation}
  G=P\,(\omega\openone+R)\,,\label{GF4}
\end{equation}
where $R\in\{0,1\}^{n\times n}$ is the symmetric graph adjacency
matrix with zeros along the diagonal, and $\openone$ is the ${n\times
  n}$ identity matrix.  The relation between the binary code
${\cal C}$ and the quaternary code $C$ is explicitly given by the
following
\begin{statement}
  The additive $\mathbb{F}_4$ code $C$ with the generator
  matrix~(\ref{GF4}) is the map $\phi(\mathscr{S})$ of the stabilizer
  $\mathscr{S}$ of the additive CWS code ${\cal Q}=({\cal G},{\cal
    C})$ generated by the graph with the adjacency matrix $R$ and the
  linear binary code ${\cal C}$ if and only if $P$ is the parity check
  matrix of ${\cal C}$.
\end{statement}
\begin{proof}
  Use the basis vectors in the form (\ref{eq:CWS-basis}) and the
  commuting operators (\ref{eq:stabilizer-operator-CWS}) corresponding
  to the rows of the matrix $G$ [Eq.~(\ref{GF4})].  Direct calculation
  gives
  \begin{equation}
    \label{eq:check-CWS}
    G_i \,Z^\mathbf{c}\ket s
    = \prod_{j=1}^n S_j^{P_{ij}} \,Z^\mathbf{c}\ket s
    = (-1)^{P_{ij}c^j} \,Z^\mathbf{c}\,\ket s,
  \end{equation}
  were summation over repeated indices is assumed.  The statement of
  the Lemma (both ways) follows from the definition of the parity
  check matrix.
\end{proof}

We can now easily relate the error detection conditions for additive
codes in Theorems \ref{Th1} (codes over $\mathbb{F}_4$) and \ref{Th2}
(CWS codes).  The code $C$ in Theorem \ref{Th1} is given by additive
combinations of the $n-k$ rows of the generator matrix~(\ref{GF4}).
Evaluating the outer trace product of the generator matrix~(\ref{GF4})
with a vector $\mathbf{e}=\mathbf{u}+\omega \mathbf{v}$, we obtain the
condition for the vector to be in $C_\perp$
\begin{equation}
  0=G*\mathbf{e}=P(\mathbf{u}+R\mathbf{v}).\label{Orthogonality}
\end{equation}
This uniform linear system of $n-k$ equations with $2n$ variables has
$n+k$ linearly-independent solutions.  The corresponding basis can be
chosen as a set of $k$ linearly-independent ``classical'' vectors
$\mathbf{e}_i=\mathbf{u}_i$, where $P \mathbf{u}_i=0$ and the
corresponding $v_i=0$, $i=1,\ldots,k$, plus $n$ linearly-independent
vectors such that $u_j=R v_j$, $j=k+1,\ldots, k+n$.  Some linear
combinations of the latter vectors are actually in $C$.  These can be
found using the identity $(C_\perp)_\perp=C$: the corresponding
$\mathbf{v}_j$ have to satisfy $\mathbf{u}_i\cdot \mathbf{v}_j=0$,
which precisely corresponds to the degenerate case, $\clg(E)=0$, in
Theorem \ref{Th2}.

General theory of CWS codes guarantees that generator matrix of a
quantum code equivalent to any additive self-orthogonal code over
$\mathbb{F}_4$ can be decomposed in the form (\ref{GF4}).
Conversely, any matrix in the form (\ref{GF4}) with binary
matrices $P$ and $R$ generates a self-orthogonal code over
$\mathbb{F}_4$ as long as the matrix $R$ is symmetric.  We use it
in the following section to prove the lower Gilbert-Varshamov (GV)
bound for the parameters of an additive CWS code which can be
obtained from a given graph.

\section{GV bound for the additive CWS codes with a given graph}
\label{sec:gv}

The GV bound is a counting argument which non-constructively proves
the existence of codes with parameters exceeding certain
threshold. The argument is based on the fact that the set of possible
codes (vector spaces) vastly outnumbers the set of vectors. Then, if
we count all codes of a given length $n$, and then subtract the number
of codes that contain any vector of weight $d-1$ or less, the
remaining codes (if any) will all have distance $d$ or more. This
``greedy'' argument ignores any possible double counting of codes that
contain several small-weight vectors.  Note that the GV bound
necessarily gives asymptotically \emph{good} codes with relative
distance $\delta\equiv d/n$ and code rate $R\equiv k/n$ bounded away
from $0$ as $n\to\infty$.

For the entire class of pure stabilizer codes,  the asymptotic GV
bound\cite{Feng:dec.2004} states that there exist
 codes of length $n\to\infty$ such that
\begin{equation}
  \delta \log_23+H_2(\delta)\ge 1-R,
\label{eq:GV1}
\end{equation}
where $H_2(\delta)\equiv -\delta\log_2
\delta-(1-\delta)\log_2(1-\delta)$ is the binary entropy function.  We
are going to prove that the same bound also holds for pure CWS codes
corresponding to a given graph ${\cal G}$, as long as $d\le d'({\cal
  G})$ [see Lemma \ref{lemma:upper-pure}].  We are using
Eq.~(\ref{GF4}) to parameterize the stabilizer matrices; the resulting
$\mathbb{F}_{4}$ codes are automatically self-orthogonal.  Let
$N_{n,k}^{\mathcal{G}}$ be the number of CWS codes that have length
$n$, dimension at least $k$, and correspond to a given graph
$\mathcal{G}$.  Let also $N_{\mathbf{e},n,k}^{\mathcal{G}}$ be the
number of such codes which contain a given vector
$\mathbf{e}=\mathbf{u}+\omega \mathbf{v}$,
$\wgt\mathbf{e}<d^{\prime}(\mathcal{G})$, in $C_{\perp}$ [see
Eq.~(\ref{GF4})]. The corresponding
condition~(\ref{eq:C-perp}) is given by the trace inner product
(\ref{Orthogonality}). For $\mathop{\rm wgt}\mathbf{e}<d^{\prime}%
(\mathcal{G})$, the binary vector
$\mathbf{c}\equiv\mathbf{u}+R\mathbf{v}$ is always non-zero (which
also guarantees that $\mathbf{e}\notin C$).  As a result,
$N_{\mathbf{e},n,k}^{\mathcal{G}}$ and $N_{n,k}^{\mathcal{G}}$
represent the corresponding numbers for the \emph{binary} codes.

Then the standard counting arguments \cite{MS-book} show that
\begin{equation}
(2^{n}-1)N_{\mathbf{e},n,k}^{\mathcal{G}}=(2^{k}-1)N_{n,k}^{\mathcal{G}}%
\end{equation}
Here we use the fact that each of $2^{n}-1$ vectors $\mathbf{c}$
belongs to the same number $N_{\mathbf{e},n,k}^{\mathcal{G}}$ of
binary codes; also each of $N_{n,k}^{\mathcal{G}}$ binary codes
contains $2^{k}-1$ nonzero vectors $\mathbf{c}$.  The number of
quaternary vectors of weight $s$ is $3^{s}{n\choose s}$. Thus, for any
graph $\mathcal{G}$, there exists a distance-$d$ CWS code as long as
\begin{equation}
N_{\mathbf{e},n,k}^{\mathcal{G}}\sum_{s=1}^{d-1}3^{s}{n\choose s}%
<N_{n,k}^{\mathcal{G}} \label{eq:greedy-full}%
\end{equation}
Now we see that there exist $[[n,k,d]]$ CWS codes for the
graph $\mathcal{G}$ with distance
\begin{equation}
d=\min \{d_{GV},d_{\mathrm{max}}\},
\end{equation}
where
$d_{\mathrm{max}}$ is the distance of the graph state $d^{\prime}%
(\mathcal{G})$ and
\begin{equation}
d_{GV}=\max d:\sum_{s=1}^{d-1}3^{s}{n\choose s}<{{\frac{2^{n}%
-1}{2^{k}-1}.}}%
\label{eq:SumP}%
\end{equation}

Note that thus obtained quantum codes are always pure, since the
summation in the l.h.s.\ can only be extended up to
$d_{\mathrm{max}}-1$. Apart from this latter condition,
Eq.~(\ref{eq:SumP}) is identical to the quantum Gilbert-Varshamov
bound\cite{Feng:dec.2004} for pure stabilizer codes, and takes the
asymptotic form~(\ref{eq:GV1}) as $n\rightarrow\infty$.

The exact GV bound $d\geq d_{GV}$ for pure stabilizer codes
(without the restriction on the distance) is recovered if we go
over different graphs. Indeed, the GV bound~(\ref{eq:GV1}) also
applies for the special case of $k=0$,
corresponding to stabilizer states or self-dual
codes\cite{Rains-Sloane-self-dual-codes-1998}. The GV bound on the
relative distance is monotonous in $k$ and reaches its maximum at
\begin{equation}
\delta_{k=0}\approx0.189. \label{eq:GV-zero-rate}%
\end{equation}
Then for given $n$ and $\delta<\delta_{k=0}$ one can always find a
suitable graph such that the GV bound $d\leq d_{GV}$ becomes more
restrictive than the condition $d\leq d_{\mathrm{max}}$.

In practice, graphs with large distance $d^{\prime}(\mathcal{G})$
are complicated (have too many edges). It is much easier to come
up with graph
families corresponding to a fixed graph-state distance $d^{\prime}%
(\mathcal{G})$. For $n\to\infty$, the corresponding code families
approach the maximum rate $R=1$ and have  asymptotic redundancy
$r\equiv n-k$ defined by Eq.~(\ref{eq:SumP}):
\begin{equation}
r\leq d\log_{2}3+nH_{2}(d/n).
\label{eq:GV-redundancy}%
\end{equation}
It is readily verified that the r.h.s of estimate
(\ref{eq:GV-redundancy}) has the order of  $d\log_2(3n/d)$ if
$d$=const and $n\to\infty$.

  \begin{figure}
    \includegraphics[width=0.5\columnwidth]{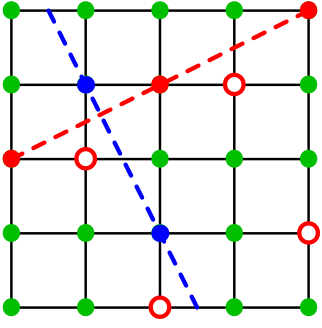}
    \caption{Square-lattice additive CWS code $[[25,4,4]]$. Circles
      represent the qubits. All $k=4$ translations of the empty-circle
      pattern form the classical codewords $\mathbf{c}_{i}$.  The
      generators of the stabilizer are formed as the products of the
      graph-state stabilizer generators along the directions parallel
      to the dashed lines.}
\label{fig3}
\end{figure}

\begin{example}
  \label{Ex:square-planar}
  Graphs in the form of sufficiently large finite square lattice fragments
  [Fig.~\ref{fig3}] have maximum distance $d_\mathrm{max}=5$, but this
  requires that the bits in the corners and around the perimeter be not
  involved in the classical code.  Somewhat better redundancy can be achieved
  by avoiding only the bits in the corners, which gives $d_\mathrm{max}=4$.
  Consider the family of classical codes where the codewords are obtained by
  taking all translations of the pattern shown in Fig.~\ref{fig3} with open
  symbols.  The weight of any linear combination of such codewords is at least
  4.  The lattice shown corresponds to the code $[[25,4,4]]$, while general
  $m_x\times m_y$ lattice gives the code with the parameters $[[m_x\,
  m_y,(m_x-3)(m_y-3),4]]$.  Asymptotically, the redundancy for $m_x=m_y$ is
  $n-k\approx 6 n^{1/2}$, $n\to \infty$, which is bigger than the logarithm in
  Eq.~(\ref{eq:GV-redundancy}).  However, the fraction of auxiliary qubits
 vanishes as $\propto 1/n^{1/2}$ for large $n$.  The code distance can
  be increased with higher-dimensional generalizations, e.g., we can
  generalize this construction to $D$-dimensional hypercubic lattice with $2D$
  nearest neighbors so that the distance is $d=2D$ in full analogy with the
  two dimensional case.  The corresponding redundancy will scale with the area
  of the boundary.
\end{example}

While the code in Fig.~\ref{fig3} serves as a good illustration
of the concept of lattice codes with simple stabilizer structure,
it is still far from optimal.  On the $5\times5$ square lattice we
have constructed numerically a code $[[25,9,5]]$ with weight-$7$
codewords which can be mapped into each other by translations and
rotations. This design is only one logical qubit short of the
best-known generic code $[[25,10,5]]$.
\begin{example}
  Consider graphs in the form of $L\times L$ square lattices wrapped into tori
  due to periodic boundary conditions.  For $L\ge 5$, these graphs have the
  distance $d'(\mathcal{G})=5$.  GV bound~(\ref{eq:greedy-full}) shows that
  the CWS codes with the following parameters can be obtained for these
  graphs: $[[25,4,5]]$, $[[36,13,5]]$, $[[49,24,5]]$, $[[64,38,5]]$,
  $[[81,53,5]]$,  \ldots.
\end{example}

\begin{example}
  Consider graphs in the form of $L\times L$ triangular lattices
  wrapped into tori due to periodic boundary conditions.  These graphs
  have the distance $d'(\mathcal{G})=\min(L,7)$.  GV
  bound~(\ref{eq:greedy-full}) shows that the CWS codes with the
  following parameters can be obtained for these graphs: $[[36,9,6]]$, $[[49,15,7]]$, $[[64,28,7]]$, $[[81,43,7]]$, \ldots.
\end{example}

\section{Single-generator additive cyclic codes}
\label{sec:cyclic-CWS}

Example \ref{Ex:713cws} shows that a cyclic additive code does not
necessarily preserve its symmetry when converted to CWS standard form.
By a \emph{cyclic} additive CWS code we just mean a code which is
cyclic in standard form, with a circulant graph.
For such a code, Eq.~(\ref{GF4}) can be rewritten as the generator
polynomial,
\begin{equation}
  \label{eq:GF4-polynomial}
  g(x)=p(x)[\omega+r(x)],
\end{equation}
where the polynomials $p(x)$ and $r(x)$ are binary, $p(x)$ is the
parity-check polynomial of a binary cyclic code (and therefore must divide
$x^n-1$), while $r(x)$ corresponds to a symmetric circulant matrix,
\begin{equation}
r(x^{n-1})=r(x)\;(\mod x^n-1).\label{eq:symmetric-polynomial}
\end{equation}
Any such \emph{symmetric} polynomial $r(x)$ leads to a valid
self-orthogonal additive code.  The dimension of the quantum code
corresponding to the generator polynomial (\ref{eq:GF4-polynomial}) is
$k=n-\deg p(x)$.

Previously, the additive cyclic QECCs were introduced in Theorem 14 of
Ref.~\cite{Calderbank-1997}, stating that any such code has two
generators.  A {\it single}-generator additive code described by
Eq.~(\ref{eq:GF4-polynomial}) represents a new setting, in which the
second generator is equal to zero.  This condition gives a
self-orthogonal additive code $C$ [see
  Sec.~\ref{sec:intro-stab-codes}] with no binary codewords (any
$\mathbf{e}\in C$, $\mathbf{e}\equiv\mathbf{u}+\omega \mathbf{v}$, has
$\mathbf{v}\neq 0$).

A somewhat wider class of {\em single-generator\/} cyclic additive
codes can be also defined via Eq.~(\ref{eq:GF4-polynomial}), without
requiring the symmetry (\ref{eq:symmetric-polynomial}) of $r(x)$.
Then two codes ${\cal Q}$ and ${\cal Q}'$ that have a generator
polynomial in the form (\ref{eq:GF4-polynomial}) with the same
$p(x)=p'(x)$ are equivalent if and only if
\begin{equation}
  r(x)=r'(x)\mod q(x),\label{eq:circulant-r-equivalence}
\end{equation}
where $q(x)=(x^n-1)/p(x)$ is the generator polynomial of
the binary code\cite{Calderbank-1997}.  Such a
polynomial~(\ref{eq:GF4-polynomial}) generates a self-orthogonal
$\mathbb{F}_4$ code if and only if \cite{Calderbank-1997}
\begin{equation}
  \label{eq:cyclic-orthogonality}
  p(x) p(x^{n-1}) r(x^{n-1})=   p(x) p(x^{n-1}) r(x) \,(\mod x^n-1).
\end{equation}
This guarantees self-orthogonality for any $r(x)$ as long as
\begin{equation}
  p(x)p(x^{n-1})=0\:\mod x^n-1.\label{eq:cyclic-orthogonality-guarantee}
\end{equation}
An alternative formulation of this sufficient condition is that the
corresponding generator polynomial $q(x)$ must contain no more than
one root from each pair $(\alpha, \alpha^{-1})$ of mutually conjugate
$n$\,th roots of unity, $\alpha^n=1$.  In particular, a
self-reciprocal (\emph{palindromic}) polynomial\footnote{In the
  literature such polynomials have also been called ``symmetric''.  We
  prefer to reserve this term for the
  polynomials~(\ref{eq:symmetric-polynomial}) which correspond to
  symmetric circulant matrices.  Palindromic polynomials have
  reflection symmetry with respect to their ``centers'', while
  Eq.~(\ref{eq:symmetric-polynomial}) corresponds to a symmetry with
  respect to the free term, with an implicit circulant symmetry.},
\begin{equation}
  \label{eq:palindromic}
  x^{\deg q(x)}q(1/x)=q(x),
\end{equation}
always contains roots in pairs $\alpha$ and $\alpha^{-1}$. For such
polynomials Eq.~(\ref{eq:cyclic-orthogonality-guarantee}) always fails
(including the special case of $q(x)=1+x$ which has only one root
$\alpha=\alpha^{-1}=1$).

\subsection{Single-generator cyclic codes from a binary code}
\label{sec:one-gen}
The algebraic condition~(\ref{eq:linear-cyclic-condition}) on check
polynomials for linear cyclic codes makes them simpler to implement
but also dramatically restricts their number. In particular, the
general counting approach to finding CWS codes [see
  Sec.~\ref{sec:gv}], where one first chooses a graph, and then
searches for a suitable binary code can hardly be applied to cyclic CWS
codes. Even for classical binary cyclic codes, there are no counting
arguments known to date that yield asymptotically good codes, let
alone the stronger GV bound (see Research Problem 9.2 in
Ref.~\onlinecite {MS-book}). Also, long BCH codes---one of the major
subclasses of cyclic codes---are asymptotically bad and have a slowly
declining relative distance $\delta\sim (2\ln R^{-1})/\log_2n$ for
any code rate $R$. On the other hand, binary cyclic codes often
achieve the best known parameters (exceeding the GV bound) on short
lengths $n\leq 256$. Thus, using simple cyclic codes in quantum design
can yield both good parameters and feasible implementation on the
short blocks.

To better evaluate code distance of single-generator quantum cyclic codes
(\ref{eq:GF4-polynomial}), we will modify our counting approach of
Sec.~\ref{sec:gv} and begin with a binary cyclic code. Namely, we will
fix some parity-check polynomial $p(x)$ with a desired degree $k$
among the binary factors of $x^n-1$. Then we will search for a
polynomial $r(x)$, either corresponding to a cyclic graph [see
  Eq.~(\ref{eq:symmetric-polynomial})], or satisfying the more general
orthogonality condition (\ref{eq:cyclic-orthogonality}).  However,
this transition will show that the parameters of quantum codes
generated this way strongly depend on the chosen binary code.  We will
concentrate exclusively on the binary codes with irreducible generator
polynomial $q(x)$. We will show that the distance of such a cyclic CWS
code is limited from below by the GV bound (or the variants thereof)
and from above by the distance of the classical cyclic code. Since GV
bound always produces asymptotically good codes, the parameters of our
quantum codes will be mostly limited (at least, for long blocks) by
their binary counterparts.

We begin with analyzing the condition~(\ref{Orthogonality}) for a
cyclic CWS code. For a vector $\mathbf{e}\in\mathbb{F}_4^n$ to be in
 $C_\perp,$ this condition can be
rewritten in terms of the corresponding polynomials,
\begin{equation}
  \label{eq:poly-orthogonality}
  p(x)\,\left[u(x)+r(x)v(x)\right]=0\:\mod x^{n}-1\:,
\end{equation}
where the coefficients of the (reversed for notational convenience) polynomial
$e(x^{n-1})\equiv u(x)+\omega v(x)$ are given by the components of the vector
$\mathbf{e}\in\mathbb{F}_4^n$.  Since binary $p(x)$ divides $x^n-1$, we can
rewrite this in terms of the corresponding generator polynomial
$q(x)=(x^n-1)/p(x)$ for the binary code ${\cal C}$,
\begin{equation}
  \label{eq:poly-orthogonality-two}
  u(x)+r(x)v(x)=0\:\mod q(x).
\end{equation}
Now, if $v(x)$ is mutually prime with $q(x)$,
Eq.~(\ref{eq:poly-orthogonality-two}) can be just solved for $r(x)$.  In this
case the answer is unique $[\mod q(x)]$.
On the other hand, multiple solutions for $r(x)$ are possible when $
\gcd[v(x),q(x)]\neq1$.
In this work, we avoid the complications caused in the latter
case\footnote{For polynomials $q(x)$ with multiple factors, distance
  estimates of quantum codes lead to the  estimates of weight
  spectra of classical cyclic codes which contain the code generated by $q(x)$,
  which is beyond the scope of this work.}
and only consider irreducible polynomials $q(x)$.

Overall, for any irreducible $q(x)$ and any $\mathbf{e}$ with
$\mathbf{v}\neq\mathbf{0}$ and $\wgt \mathbf{e}<d({\cal C})$,
Eq.~(\ref{eq:poly-orthogonality-two}) has a unique solution for $r(x)$
such that $\deg r(x)<\deg q(x)=n-k$. Respectively, there is no more
than one additive quantum code such that $\mathbf{e}\in
C_\perp$. Generally, only some of thus obtained $r(x)$ correspond to
self-orthogonal codes, see Eq.~(\ref{eq:cyclic-orthogonality}).

Below we complete the greedy argument by counting the polynomials
$r(x)$ corresponding to self-orthogonal codes,
Eq.~(\ref{eq:cyclic-orthogonality}).  We consider separately the case
when the irreducible polynomial $q(x)$ is palindromic [see
  Eq.~(\ref{eq:palindromic})] in Lemma~\ref{lemma:GV-cyclic} below,
and when it is not in
\begin{statement}
  \label{lemma:cyclic-one}
  Consider a cyclic binary code ${\cal C}[n,k,d_{\cal C}]$ with the
  generator polynomial $q(x)$ which is both irreducible and
  non-palindromic, $x^{\deg q(x)}q(x^{-1})\neq q(x)$.  Then there
  exists a single-generator additive cyclic code $[[n,k,d]]$ with
  distance
  $$d= \min(d_{\cal C},d_\mathrm{GV}),$$
  restricted  by both the distance $d_{\cal C}$ of the binary
  code and the following variant of GV bound
 \begin{equation}
    \label{eq:gv-cyclic-one}
 d_\mathrm{GV}=\max d:    \sum_{s=1}^{d-1}
    \left(3^s-3\right){\gcd(s,n)\over n}{n\choose s}\le 2^{n-k}-2.
    \end{equation}
\end{statement}
\begin{proof}
  The non-palindromic generator polynomial $q(x)$ is one of the
  factors in $x^n-1$ which also contains its reciprocal, $x^{\deg
    q(x)}q(x^{-1})$.  This implies that the corresponding parity-check
  polynomial $p(x)$ satisfies Eq.~(\ref{eq:cyclic-orthogonality}).
  Further, since $q(x)$ is irreducible, the solution $r(x)$ of
  Eq.~(\ref{eq:poly-orthogonality-two}) is unique assuming $v(x)\neq
  0$ and $\deg r(x)<n-k$, which gives the exponential term in the
  r.h.s.\ of Eq.~(\ref{eq:gv-cyclic-one}).
  Eq.~(\ref{eq:gv-cyclic-one}) improves on the standard GV
  inequality~(\ref{eq:SumP}) by discarding a few sets of vectors.  The
  first set are vectors with $u(x)=0$, which implies $r(x)=0\:\mod
  q(x)$.  The second set are vectors with $u(x)=v(x)$, which all give
  $r(x)=1\:\mod q(x)$.  The third set are non-zero vectors with
  $v(x)=0$, which can never be in $C_\perp$ corresponding to the
  generator~(\ref{eq:GF4-polynomial}). Finally, note that any error
  vector of weight $s$ produces at least $n/\gcd(s,n)$ different
  cyclic shifts. All of these cyclic shifts give the same polynomial
  $r(x)$ and can be discounted.  The condition $d\le d_{\cal C}$ comes
  from Lemma \ref{lemma:upper-binary}.
\end{proof}

\textit{Note.} The lhs of bound (\ref{eq:gv-cyclic-one}) limits
the number of cyclic classes for $4$-ary vectors $\mathbf{e}$ of
weight $s\leq d-1.$ Most of these vectors have the maximum period
$n.$ Therefore, it can be proven that the term $\gcd(s,n)/n$ in
(\ref{eq:gv-cyclic-one}) can be replaced with the smaller term
that rapidly tends to 1/n for large n.  In turn, bound
(\ref{eq:gv-cyclic-one}) adds about $\log_{2}n$ information qubits
to bound (\ref{eq:SumP}) but tends to the standard quantum GV
bound~(\ref{eq:GV1}) as $n\rightarrow\infty$.

In the following example, this bound coincides with inequality
$d\leq d_{\mathcal{C}}$, which uniquely sets code distance $d$.

\begin{example}
  The family of the binary codes with the parameters
  $[n=2^{4h}+2^{3h}-2^h-1,k=n-6h,3]$, $h=1,2,\ldots$, constructed in
  Ref.~\onlinecite{Ding:OCT2002}, has irreducible non-palindromic
  generator polynomials as required in Lemma~\ref{lemma:cyclic-one}.
  For $d_{\cal C}=3$ the sum in Eq.~(\ref{eq:gv-cyclic-one}) has just
  one term at $s=2$; for $n$ odd $\gcd(n,2)=1$.  Explicit calculation
  confirms that the parameters of these codes satisfy inequality
  (\ref{eq:gv-cyclic-one}), which proves the existence of
  single-generator cyclic quantum codes with exactly the same parameters,
  $[[n=2^{4h}+2^{3h}-2^h-1,k=n-6h,3]]$, but not necessarily cyclic CWS codes.
  The smallest of these codes, $[[21,15,3]]$, corresponds to a
  polynomial $q(x)=1+x+x^2+x^4+x^6$ (unique up to a reversal) and can
  be obtained from an order-$21$ circulant graph corresponding to
  $r(x)=x+x^4+x^{17}+x^{20}$.  This particular combination of
  parameters gives the best existing
  code\cite{Grassl:codetables}.
\end{example}

\begin{example}
  According to the BCH bound \cite{MS-book}, a cyclic code has
  distance $d\ge r+1$ ($r+1$ is the ``{\it designed\/}'' distance) if
  the corresponding generator polynomial $q(x)$ has $r$ consecutive
  roots, e.g., $\alpha, \alpha^2, \ldots, \alpha^r$, where $\alpha$ is
  the primitive $n\,$th root of unity. A polynomial $m_\alpha(x)$
  which has root $\alpha$, necessarily has $s$ distinct roots
  $\alpha^{2^j}$ for all $j=0,...,s-1$ if $s$ is the smallest number
  such that $2^s=1 \mod n$. We say that the code has zeros
  $\alpha^{i}$, where exponents $i$ form the set $I=\{2^j (\mod n),\;
  j=0,...,s-1\}.$ The code generated by $m_\alpha(x)$ has designed
  distance $5$ if $3\in I$ or, equivalently, if $2^s=3\mod n$ for some
  $s$. The polynomial $m_\alpha(x)$ is non-palindromic if $-1\not\in
  I$. We can further obtain codes with irreducible non-palindromic
  generators and designed distances $7$, $9$, etc., by imposing
  additional conditions, e.g., $5\in I$, $7\in I$, etc. The values of
  $n,$ for which this is possible form an infinite set, $\{23_5, 47_5,
  71_7, 95_5, 115_5, 143_5, 167_5, 191_7,235_5,239_7,\ldots\}$, where
  the subscripts indicate the designed distances.  In fact, the first
  three codes represent the well known quadratic-residue codes with
  the higher distances (exceeding the BCH bound) equal to 7, 11, and
  11, respectively.  GV bound proves the existence of additive quantum
  CWS codes with the parameters $[[23,12,4]]$, $[[47,24,d\ge 6]]$,
  $[[71,36,d\ge 7]]$, $[[95,59,5]]$, $[[115,71,5]]$, $[[143,83,11]]$,
  etc.  The first three codes have the parameters as good as any known
  codes with such $n$ and $k$.
\end{example}

Now let us consider the case of a palindromic polynomial $q(x)$.  First,
we prove
\begin{statement}
  \label{lemma:symmetric-q-cws}
  Consider a binary code ${\cal C}$ generated by a palindromic
  polynomial $q(x)$
  such that $q(1)=1$.  Then any quantum
  code~(\ref{eq:GF4-polynomial}) which satisfies self-orthogonality
  condition~(\ref{eq:cyclic-orthogonality}) is equivalent to a cyclic
  CWS code with a symmetric polynomial $r(x)$, see
  Eq.~(\ref{eq:symmetric-polynomial}).
\end{statement}
\begin{proof}
  The corresponding check polynomial $p(x)$ is symmetric, thus the
  condition~(\ref{eq:cyclic-orthogonality}) can be rewritten as
  \begin{equation}
    \label{eq:cyclic-orthogonality-symmetric}
    r(x)+r(x^{n-1})= 0\mod q(x).
  \end{equation}
  The condition $q(1)=1$ guarantees that the palindromic polynomial
  $q(x)$ has odd weight and even degree $2m$, in which case the
  ``central'' monomial $x^m$ has non-zero coefficient $q_m=1$.  Given
  the block length $n$, let us choose an equivalent code [see
  Eq.~(\ref{eq:circulant-r-equivalence})] with $r(x)$ such that the
  coefficients
  \begin{equation}
    r_{m+1}=r_{m+2}=\ldots =r_{n-m}=0.
    \label{eq:zero-set}
  \end{equation}
The coefficients of the polynomial in the l.h.s.\ of
Eq.~(\ref{eq:cyclic-orthogonality-symmetric}) satisfy the same
condition, except for the term $x^{n-m}$, which has coefficient
$r_m$.  The coefficients are arranged in such a way that the
l.h.s.\ of Eq.~(\ref{eq:cyclic-orthogonality-symmetric}) can only
equal zero or $x^{n-m} q(x)\mod x^n-1$.  However, the latter
possibility can be excluded by comparing the corresponding
coefficients  of the free term $x^0$.  The only remaining case
corresponds to a symmetric $r(x)$.
\end{proof}

It is now clear that for a palindromic irreducible generator
polynomial $q(x)\neq 1+x$, one should reduce the count in the r.h.s.\
of Eq.~(\ref{eq:gv-cyclic-one}) by replacing $2^{n-k}$ with
$2^{(n-k)/2}$, the number of symmetric polynomials that satisfy
Eq.~(\ref{eq:zero-set}).  This gives the following version of GV bound,
\begin{equation}
d_{\mathrm{GV}}=\max d:
\sum_{s=1}^{d-1}\left(3^s-3\right){\gcd(s,n)\over n}{n\choose s}\le
2^{(n-k)/2}-2.\label{eq:gv-cyclic-three}
\end{equation}
While the resulting estimate is much weaker than the GV
bound~(\ref{eq:gv-cyclic-one}), it still gives asymptotically good
codes.  A better (especially, for small $d$) bound is given in the
following
\begin{statement}
  Consider a cyclic binary code ${\cal C}[n,k,d_{\cal C}]$ with
  $d_{\cal C}\ge3$ and the generator polynomial $q(x)$ which is both
  palindromic and irreducible.  Then there exists a cyclic CWS code
  $[[n,k,d]]$ with the distance $d=\min(d_{\cal C},\lfloor
  d_\mathrm{GV}/2\rfloor)$, where $d_\mathrm{GV}=\max d$:
  \begin{equation}
    \label{eq:gv-cyclic-two}
\sum_{s=1}^{d-1}(3^{\lceil s/2\rceil}-3){\gcd(s,n)\over n}
{\lfloor n/2\rfloor\choose \lfloor s/2\rfloor}\le 2^{
(n-k)/2}-2.
\end{equation}
\label{lemma:GV-cyclic}
\end{statement}
\begin{proof}
  The restriction on the distance guarantees that $q(x)\neq 1+x$, and
  therefore $q(x)$ satisfies the conditions of
  Lemma~\ref{lemma:symmetric-q-cws}; in particular, $n-k$ is even.
  The inequality just corresponds to symmetric polynomials $r(x)$ and
  the errors that are also symmetric, $e(x^{n-1})=e(x)\mod x^n-1$.
  The statement of the Lemma follows from the fact that for any
  general error $e(x)$, there is a symmetric error
  $e_\mathrm{sym}(x)\equiv e(x)+e(x^{n-1})\mod x^n-1$ whose weight is
  even and is limited by $\wgt e_\mathrm{sym}(x)\le 2 \wgt e(x)$.
\end{proof}



\begin{example} Among classical codes, the largest distance is obtained for
  the repetition codes, with the parameters ${\cal C}=[n,1,n]$.  The
  parity-check polynomial is $p(x)=x+1$; the generator polynomial
  $q(x)=1+x+\ldots+x^{n-1}$ is irreducible and palindromic for $n=2$, and for
  all $n>2$ that satisfy the condition $\ord_m(2)=m-1$, where $\ord_m(q)$ is
  the multiplicative order of $q$ modulo $m$.  This includes the following
  $n\le 100$:
  \begin{equation}
     \{3, 5, 11, 13, 19, 29, 37, 53, 59, 61, 67, 83,
    \ldots\}.
    \label{eq:prime-set}
\end{equation}
Lemma~\ref{lemma:GV-cyclic} shows that for $n$ from the
set~(\ref{eq:prime-set}), additive cyclic CWS codes with parameters
$[[n,1,\lfloor d_{GV}(n,1)/2\rfloor]]$ exist, where $d_{GV}(n,1)$ is
obtained from Eq.~(\ref{eq:gv-cyclic-two}) with $k=1$. Asymptotically,
at large $n$, this corresponds to cyclic codes with the relative
distance given by half of that given by Eq.~(\ref{eq:GV-zero-rate}).
\end{example}

\begin{example}
  (Cyclic analogs of the toric code) In the setting of the previous example,
  cyclic CWS codes $[[5,1,3]]$, $[[13,1,5]]$, $[[25,1,7]]$, $[[41,1,9]]$ with
  $p(x)=1+x$ were obtained numerically.  The corresponding graph-state
  generators are $ZXZ$ for $d=3$, $ZZXZZ$ for $d=5$, etc.  We obtain a family
  of cyclic codes with the weight-$4$ generators, $S_3=ZXXZ$, $S_5=ZIXXIZ$,
  etc.  Codes with generators $S_3=ZXIXZ$, $S_5=ZXIIIXZ$, $S_7=ZXIIIIIXZ$, and
  $S_9=ZXIIIIIIIXZ$ have the same parameters (the corresponding graphs are
  somewhat more complicated).  The latter family can be generalized to codes
  with $n=t^2+(t+1)^2$, $k=1$, $d=2t+1$, $t=1,2,\ldots$; the corresponding
  stabilizer generators $S_{2t+1}$ having $2t-1$ identity operators separating
  $ZX$ and $XZ$.  These cyclic codes are related to generalized toric
  codes\cite{Bravyi-1998,Freedman-1998,Bombin-2007}; the square-lattice qubit
  layout preserving the circulant symmetry is illustrated in
  Fig.~\ref{fig:skewed} for $t=1,2$.
  \label{ex:toric-cws}
\end{example}

\begin{figure}[htbp]
  \centering
  \includegraphics[width=0.36\columnwidth]{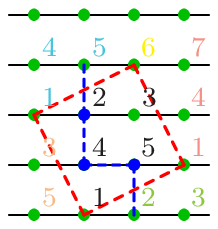}\
  \includegraphics[width=0.54\columnwidth]{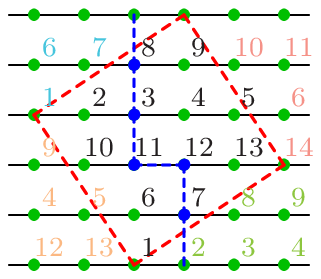}
  \caption{(Color online) Left: Correspondence between the cyclic code
    $[[5,1,3]]$ with generators $ZXIXZ$, and a generalized toric code
    on square lattice.  Only qubits within the dashed square are
    participating in the code; periodic boundary conditions are
    assumed.  The vertical dashed line indicates a topologically
    nontrivial chain of errors which limits the distance of the code:
    $X_2Z_4Z_5$ is equivalent to $Z_1Z_2Z_3Z_4Z_5$.  Right: same for
    the code $[[13,1,5]]$ with the generator $ZXIIIXZ$.}
  \label{fig:skewed}
\end{figure}

\begin{example} (CWS codes from $k$ copies of the repetition code) Let
  us take the binary code ${\cal C}$ to be formed by $k$ copies of the
  repetition code with the distance $d_2=m$.  Then the block size is
  $n=km$, and the check polynomial is $p(x)=x^k-1$.  The generator
  polynomial $q(x)=1+x^k+\ldots+x^{k(m-1)}$ is always palindromic; it
  is also irreducible if $m$ belongs to the set (\ref{eq:prime-set}),
  while $k=m^s$, $s=0,1, 2, \ldots$.  For sufficiently large $n$,
  Lemma~\ref{lemma:GV-cyclic} gives asymptotically good codes with
  $kd\propto n^2$.  Since these parameters cannot exceed those of the
  binary code ${\cal C}$ which correspond to $kd=n$ (thus
  $\delta=1/k$), for these values of $m$ and $k> 10$, there exist
  cyclic CWS codes with the parameters of the corresponding binary
  code, $[[n=m^{s+1},m^s,d=m]]$.  This prediction is readily verified
  empirically, see Table~\ref{tab:quantum-repetition}.  Note that, as
  in the Example \ref{ex:toric-cws}, many of these codes have
  stabilizer generators with small weight.
    \label{ex:k-copies}
\end{example}

\begin{table}[htbp]
  \centering
  \begin{tabular}[c]{c|c|c|c|c|c|c}
 $m$&$n$&$k$&$d$&$S_1$\\
 \hline
 3 &  6 & 2 & 2 &  $ZXZ$ \\
  &  9 & 3 & 3 &  $ZXZ$\\
  &  12 & 3 & 3 &  $ZXZ$\\  \hline
  5 &   5 & 1 & 3 &  $ZXZ$ \\
  & 10 & 2 & 3 &  $ZXZ$\\
  & 15 & 3 & 5 &  $ZIZIIXIIZIZ$\\
  & 20 & 4 & 5 &  $ZIZZXZZIZ$\\
  & 25 & 5 & 5 &  $ZIZZXZZIZ$\\ \hline
7  & 7  & 1 & 3 &  $ZXZ$\\
  & 14 & 2 & 5 &  $ZZIXIZZ$ \\
  & 21 & 3 & 6 &  $ZIZZXZZIZ$ \\
  & 28 & 4 & 7 &  $ZIZZIXIZZIZ$\\ \hline
  \end{tabular}
  \caption{Families of the cyclic codes obtained numerically from $k$ copies 
    of the
    classical repetition code, $p(x)=x^m-1$, corresponding to
    $m=3,5,7$.   The expected distance saturation,
    $d=m$, is reached already at  $k\le4$, even for $k$ and $m$ different
    from those in Example~\ref{ex:k-copies}.
    The operator strings in the last column are
    representative graph-state generators (the
    remaining generators are obtained by cyclic shifts).}
  \label{tab:quantum-repetition}
\end{table}

\section{Conclusions}
\label{sec:conclusions}

In this paper, we analyze how  the general CWS framework can
facilitate the search of the additive quantum codes with
reasonably good parameters. Unlike  complete optimization of CWS
codes\cite{Chuang-CWS-2009}, which involves all non-isomorphic
LC-inequivalent order-$n$ graphs and all binary codes of length
$n$, here one can independently pick a suitable graph ${\cal G}$,
and then search for a linear binary code ${\cal C}$ that can
correct the error patterns induced by ${\cal G}$, see
Eq.~(\ref{Orthogonality}).

The choice of the graph is discussed in
Sec.~\ref{sec:upper-bound}. In the simplest case of pure codes,
one has to pick a graph with a sufficiently large graph-state
distance $d'({\cal G})$.  Assuming that a regular graph is being
sought in this design, we consider graphs with minimal vertex
degrees $d'({\cal G})-1$ or more.

After the graph is chosen, the second step involves the search of
an appropriate binary code. Here in Sec.~\ref{sec:gv} we prove the
existence of the binary codes that give good quantum CWS codes.
The corresponding lower bound on the distance is given by the
quantum Gilbert-Varshamov bound~(\ref{eq:SumP}). Note that while
this bound is proved for a given graph, it is the same bound that
holds for a generic stabilizer code.

Our results show that  by restricting the graph ${\cal G}$ of a
CWS code to regular lattices, one can lower the complexity of the
code search and still obtain codes with relatively good
parameters. On the other hand, the graph structure could be mapped
directly to a physical qubit layout. Therefore, such codes can
simplify both the hardware design and the error-correcting
procedures, which will easily admit the property of translational
invariance.

An unexpected byproduct of this work is the discovery of a previously
unexplored family of single-generator quantum cyclic codes
(Sec.~\ref{sec:one-gen}).  These codes are relatively easy to
construct, and they are plentiful.  We construct (or prove the
existence) of several simple families of such codes that have
unbounded distances.  These include cyclic CWS codes with weight-$4$
stabilizer generators, which turned out to be toric codes in disguise
(Example \ref{ex:toric-cws}), as well as a code family with the
parameters of generalized repetition codes, $[[kd,k,d]]$ (Example
\ref{ex:k-copies}).  The main advantage of these families is a simple
structure of their stabilizers.

\section{Acknowledgments}
We are grateful to M.~Grassl for helpful comments.  This work was supported in
part by the Army Research Office 
through the grant No.\ W911NF-11-1-0027 (LP \& ID), and by NSF through
the grant No.\ 1018935 (LP).


\begin{thebibliography}{57}%
\makeatletter
\providecommand \@ifxundefined [1]{%
 \@ifx{#1\undefined}
}%
\providecommand \@ifnum [1]{%
 \ifnum #1\expandafter \@firstoftwo
 \else \expandafter \@secondoftwo
 \fi
}%
\providecommand \@ifx [1]{%
 \ifx #1\expandafter \@firstoftwo
 \else \expandafter \@secondoftwo
 \fi
}%
\providecommand \natexlab [1]{#1}%
\providecommand \enquote  [1]{``#1''}%
\providecommand \bibnamefont  [1]{#1}%
\providecommand \bibfnamefont [1]{#1}%
\providecommand \citenamefont [1]{#1}%
\providecommand \href@noop [0]{\@secondoftwo}%
\providecommand \href [0]{\begingroup \@sanitize@url \@href}%
\providecommand \@href[1]{\@@startlink{#1}\@@href}%
\providecommand \@@href[1]{\endgroup#1\@@endlink}%
\providecommand \@sanitize@url [0]{\catcode `\\12\catcode `\$12\catcode
  `\&12\catcode `\#12\catcode `\^12\catcode `\_12\catcode `\%12\relax}%
\providecommand \@@startlink[1]{}%
\providecommand \@@endlink[0]{}%
\providecommand \url  [0]{\begingroup\@sanitize@url \@url }%
\providecommand \@url [1]{\endgroup\@href {#1}{\urlprefix }}%
\providecommand \urlprefix  [0]{URL }%
\providecommand \Eprint [0]{\href }%
\providecommand \doibase [0]{http://dx.doi.org/}%
\providecommand \selectlanguage [0]{\@gobble}%
\providecommand \bibinfo  [0]{\@secondoftwo}%
\providecommand \bibfield  [0]{\@secondoftwo}%
\providecommand \translation [1]{[#1]}%
\providecommand \BibitemOpen [0]{}%
\providecommand \bibitemStop [0]{}%
\providecommand \bibitemNoStop [0]{.\EOS\space}%
\providecommand \EOS [0]{\spacefactor3000\relax}%
\providecommand \BibitemShut  [1]{\csname bibitem#1\endcsname}%
\let\auto@bib@innerbib\@empty
\bibitem [{\citenamefont {Shor}(1995)}]{shor-error-correct}%
  \BibitemOpen
  \bibfield  {author} {\bibinfo {author} {\bibfnamefont {P.~W.}\ \bibnamefont
  {Shor}},\ }\href {http://link.aps.org/abstract/PRA/v52/pR2493} {\bibfield
  {journal} {\bibinfo  {journal} {Phys. Rev. A}\ }\textbf {\bibinfo {volume}
  {52}},\ \bibinfo {pages} {R2493} (\bibinfo {year} {1995})}\BibitemShut
  {NoStop}%
\bibitem [{\citenamefont {Knill}\ and\ \citenamefont
  {Laflamme}(1997)}]{Knill-Laflamme-1997}%
  \BibitemOpen
  \bibfield  {author} {\bibinfo {author} {\bibfnamefont {E.}~\bibnamefont
  {Knill}}\ and\ \bibinfo {author} {\bibfnamefont {R.}~\bibnamefont
  {Laflamme}},\ }\href {http://dx.doi.org/10.1103/PhysRevA.55.900} {\bibfield
  {journal} {\bibinfo  {journal} {Phys. Rev. A}\ }\textbf {\bibinfo {volume}
  {55}},\ \bibinfo {pages} {900} (\bibinfo {year} {1997})}\BibitemShut
  {NoStop}%
\bibitem [{\citenamefont {Bennett}\ \emph {et~al.}(1996)\citenamefont
  {Bennett}, \citenamefont {DiVincenzo}, \citenamefont {Smolin},\ and\
  \citenamefont {Wootters}}]{Bennett-1996}%
  \BibitemOpen
  \bibfield  {author} {\bibinfo {author} {\bibfnamefont {C.}~\bibnamefont
  {Bennett}}, \bibinfo {author} {\bibfnamefont {D.}~\bibnamefont {DiVincenzo}},
  \bibinfo {author} {\bibfnamefont {J.}~\bibnamefont {Smolin}}, \ and\ \bibinfo
  {author} {\bibfnamefont {W.}~\bibnamefont {Wootters}},\ }\href
  {http://dx.doi.org/10.1103/PhysRevA.54.3824} {\bibfield  {journal} {\bibinfo
  {journal} {Phys. Rev. A}\ }\textbf {\bibinfo {volume} {54}},\ \bibinfo
  {pages} {3824} (\bibinfo {year} {1996})}\BibitemShut {NoStop}%
\bibitem [{\citenamefont {Knill}\ \emph {et~al.}(1998)\citenamefont {Knill},
  \citenamefont {Laflamme},\ and\ \citenamefont {Zurek}}]{Knill-error-bound}%
  \BibitemOpen
  \bibfield  {author} {\bibinfo {author} {\bibfnamefont {E.}~\bibnamefont
  {Knill}}, \bibinfo {author} {\bibfnamefont {R.}~\bibnamefont {Laflamme}}, \
  and\ \bibinfo {author} {\bibfnamefont {W.~H.}\ \bibnamefont {Zurek}},\ }\href
  {http://www.sciencemag.org/cgi/content/abstract/279/5349/342} {\bibfield
  {journal} {\bibinfo  {journal} {Science}\ }\textbf {\bibinfo {volume}
  {279}},\ \bibinfo {pages} {342} (\bibinfo {year} {1998})}\BibitemShut
  {NoStop}%
\bibitem [{\citenamefont {Rahn}\ \emph {et~al.}(2002)\citenamefont {Rahn},
  \citenamefont {Doherty},\ and\ \citenamefont {Mabuchi}}]{Rahn-2002}%
  \BibitemOpen
  \bibfield  {author} {\bibinfo {author} {\bibfnamefont {B.}~\bibnamefont
  {Rahn}}, \bibinfo {author} {\bibfnamefont {A.~C.}\ \bibnamefont {Doherty}}, \
  and\ \bibinfo {author} {\bibfnamefont {H.}~\bibnamefont {Mabuchi}},\ }\href
  {http://dx.doi.org/10.1103/PhysRevA.66.032304} {\bibfield  {journal}
  {\bibinfo  {journal} {Phys. Rev. A}\ }\textbf {\bibinfo {volume} {66}},\
  \bibinfo {pages} {032304} (\bibinfo {year} {2002})}\BibitemShut {NoStop}%
\bibitem [{\citenamefont {Dennis}\ \emph {et~al.}(2002)\citenamefont {Dennis},
  \citenamefont {Kitaev}, \citenamefont {Landahl},\ and\ \citenamefont
  {Preskill}}]{Dennis-Kitaev-Landahl-Preskill-2002}%
  \BibitemOpen
  \bibfield  {author} {\bibinfo {author} {\bibfnamefont {E.}~\bibnamefont
  {Dennis}}, \bibinfo {author} {\bibfnamefont {A.}~\bibnamefont {Kitaev}},
  \bibinfo {author} {\bibfnamefont {A.}~\bibnamefont {Landahl}}, \ and\
  \bibinfo {author} {\bibfnamefont {J.}~\bibnamefont {Preskill}},\ }\href
  {http://dx.doi.org/10.1063/1.1499754} {\bibfield  {journal} {\bibinfo
  {journal} {J. Math. Phys.}\ }\textbf {\bibinfo {volume} {43}},\ \bibinfo
  {pages} {4452} (\bibinfo {year} {2002})}\BibitemShut {NoStop}%
\bibitem [{\citenamefont {Steane}(2003)}]{Steane-2003}%
  \BibitemOpen
  \bibfield  {author} {\bibinfo {author} {\bibfnamefont {A.~M.}\ \bibnamefont
  {Steane}},\ }\href {http://dx.doi.org/10.1103/PhysRevA.68.042322} {\bibfield
  {journal} {\bibinfo  {journal} {Phys. Rev. A}\ }\textbf {\bibinfo {volume}
  {68}},\ \bibinfo {pages} {042322} (\bibinfo {year} {2003})}\BibitemShut
  {NoStop}%
\bibitem [{\citenamefont {Fowler}\ \emph
  {et~al.}(2004{\natexlab{a}})\citenamefont {Fowler}, \citenamefont {Hill},\
  and\ \citenamefont {Hollenberg}}]{Fowler-QEC-2004}%
  \BibitemOpen
  \bibfield  {author} {\bibinfo {author} {\bibfnamefont {A.~G.}\ \bibnamefont
  {Fowler}}, \bibinfo {author} {\bibfnamefont {C.~D.}\ \bibnamefont {Hill}}, \
  and\ \bibinfo {author} {\bibfnamefont {L.~C.~L.}\ \bibnamefont
  {Hollenberg}},\ }\href {http://link.aps.org/abstract/PRA/v69/e042314}
  {\bibfield  {journal} {\bibinfo  {journal} {Phys. Rev. A}\ }\textbf {\bibinfo
  {volume} {69}},\ \bibinfo {pages} {042314} (\bibinfo {year}
  {2004}{\natexlab{a}})}\BibitemShut {NoStop}%
\bibitem [{\citenamefont {Fowler}\ \emph
  {et~al.}(2004{\natexlab{b}})\citenamefont {Fowler}, \citenamefont {Devitt},\
  and\ \citenamefont {Hollenberg}}]{Fowler-2005}%
  \BibitemOpen
  \bibfield  {author} {\bibinfo {author} {\bibfnamefont {A.~G.}\ \bibnamefont
  {Fowler}}, \bibinfo {author} {\bibfnamefont {S.~J.}\ \bibnamefont {Devitt}},
  \ and\ \bibinfo {author} {\bibfnamefont {L.~C.~L.}\ \bibnamefont
  {Hollenberg}},\ }\href {http://arxiv.org/abs/quant-ph/0402196} {\bibfield
  {journal} {\bibinfo  {journal} {Quant. Info. Comput.}\ }\textbf {\bibinfo
  {volume} {4}},\ \bibinfo {pages} {237} (\bibinfo {year}
  {2004}{\natexlab{b}})},\ \bibinfo {note} {quant-ph/0402196}\BibitemShut
  {NoStop}%
\bibitem [{\citenamefont {{Fowler}}(2005)}]{fowler-thesis-2005}%
  \BibitemOpen
  \bibfield  {author} {\bibinfo {author} {\bibfnamefont {A.~G.}\ \bibnamefont
  {{Fowler}}},\ }\href {http://arxiv.org/abs/quant-ph/0506126} {\enquote
  {\bibinfo {title} {{Towards Large-Scale Quantum Computation}},}\ } (\bibinfo
  {year} {2005}),\ \bibinfo {note} {arXiv:quant-ph/0506126}\BibitemShut
  {NoStop}%
\bibitem [{\citenamefont {Knill}(2005{\natexlab{a}})}]{Knill-nature-2005}%
  \BibitemOpen
  \bibfield  {author} {\bibinfo {author} {\bibfnamefont {E.}~\bibnamefont
  {Knill}},\ }\href {http://dx.doi.org/10.1038/nature03350} {\bibfield
  {journal} {\bibinfo  {journal} {Nature}\ }\textbf {\bibinfo {volume} {434}},\
  \bibinfo {pages} {39} (\bibinfo {year} {2005}{\natexlab{a}})}\BibitemShut
  {NoStop}%
\bibitem [{\citenamefont {Knill}(2005{\natexlab{b}})}]{Knill-2005}%
  \BibitemOpen
  \bibfield  {author} {\bibinfo {author} {\bibfnamefont {E.}~\bibnamefont
  {Knill}},\ }\href {http://dx.doi.org/10.1103/PhysRevA.71.042322} {\bibfield
  {journal} {\bibinfo  {journal} {Phys. Rev. A}\ }\textbf {\bibinfo {volume}
  {71}},\ \bibinfo {pages} {042322} (\bibinfo {year}
  {2005}{\natexlab{b}})}\BibitemShut {NoStop}%
\bibitem [{\citenamefont {Raussendorf}\ and\ \citenamefont
  {Harrington}(2007)}]{Raussendorf-Harrington-2007}%
  \BibitemOpen
  \bibfield  {author} {\bibinfo {author} {\bibfnamefont {R.}~\bibnamefont
  {Raussendorf}}\ and\ \bibinfo {author} {\bibfnamefont {J.}~\bibnamefont
  {Harrington}},\ }\href {http://link.aps.org/abstract/PRL/v98/e190504}
  {\bibfield  {journal} {\bibinfo  {journal} {Phys. Rev. Lett.}\ }\textbf
  {\bibinfo {volume} {98}},\ \bibinfo {pages} {190504} (\bibinfo {year}
  {2007})}\BibitemShut {NoStop}%
\bibitem [{\citenamefont {Vandersypen}\ \emph {et~al.}(2000)\citenamefont
  {Vandersypen}, \citenamefont {Steffen}, \citenamefont {Breyta}, \citenamefont
  {Yannoni}, \citenamefont {Cleve},\ and\ \citenamefont
  {Chuang}}]{chuang-2000}%
  \BibitemOpen
  \bibfield  {author} {\bibinfo {author} {\bibfnamefont {L.~M.~K.}\
  \bibnamefont {Vandersypen}}, \bibinfo {author} {\bibfnamefont
  {M.}~\bibnamefont {Steffen}}, \bibinfo {author} {\bibfnamefont
  {G.}~\bibnamefont {Breyta}}, \bibinfo {author} {\bibfnamefont {C.~S.}\
  \bibnamefont {Yannoni}}, \bibinfo {author} {\bibfnamefont {R.}~\bibnamefont
  {Cleve}}, \ and\ \bibinfo {author} {\bibfnamefont {I.~L.}\ \bibnamefont
  {Chuang}},\ }\href {http://link.aps.org/abstract/PRL/v85/p5452} {\bibfield
  {journal} {\bibinfo  {journal} {Phys. Rev. Lett.}\ }\textbf {\bibinfo
  {volume} {85}},\ \bibinfo {pages} {5452} (\bibinfo {year}
  {2000})}\BibitemShut {NoStop}%
\bibitem [{\citenamefont {Vandersypen}\ \emph {et~al.}(2001)\citenamefont
  {Vandersypen}, \citenamefont {Steffen}, \citenamefont {Breyta}, \citenamefont
  {Yannoni}, \citenamefont {Sherwood},\ and\ \citenamefont
  {Chuang}}]{chuang-2001}%
  \BibitemOpen
  \bibfield  {author} {\bibinfo {author} {\bibfnamefont {L.~M.~K.}\
  \bibnamefont {Vandersypen}}, \bibinfo {author} {\bibfnamefont
  {M.}~\bibnamefont {Steffen}}, \bibinfo {author} {\bibfnamefont
  {G.}~\bibnamefont {Breyta}}, \bibinfo {author} {\bibfnamefont {C.~S.}\
  \bibnamefont {Yannoni}}, \bibinfo {author} {\bibfnamefont {M.~H.}\
  \bibnamefont {Sherwood}}, \ and\ \bibinfo {author} {\bibfnamefont {I.~L.}\
  \bibnamefont {Chuang}},\ }\href {http://dx.doi.org/10.1038/414883a}
  {\bibfield  {journal} {\bibinfo  {journal} {Nature}\ }\textbf {\bibinfo
  {volume} {414}},\ \bibinfo {pages} {883} (\bibinfo {year}
  {2001})}\BibitemShut {NoStop}%
\bibitem [{\citenamefont {Gulde}\ \emph {et~al.}(2003)\citenamefont {Gulde},
  \citenamefont {Riebe}, \citenamefont {Lancaster}, \citenamefont {Becher},
  \citenamefont {Eschner}, \citenamefont {H\"affner}, \citenamefont
  {Schmidt-Kaler}, \citenamefont {Chuang},\ and\ \citenamefont
  {Blatt}}]{Gulde-2003}%
  \BibitemOpen
  \bibfield  {author} {\bibinfo {author} {\bibfnamefont {S.}~\bibnamefont
  {Gulde}}, \bibinfo {author} {\bibfnamefont {M.}~\bibnamefont {Riebe}},
  \bibinfo {author} {\bibfnamefont {G.~P.~T.}\ \bibnamefont {Lancaster}},
  \bibinfo {author} {\bibfnamefont {C.}~\bibnamefont {Becher}}, \bibinfo
  {author} {\bibfnamefont {J.}~\bibnamefont {Eschner}}, \bibinfo {author}
  {\bibfnamefont {H.}~\bibnamefont {H\"affner}}, \bibinfo {author}
  {\bibfnamefont {F.}~\bibnamefont {Schmidt-Kaler}}, \bibinfo {author}
  {\bibfnamefont {I.~L.}\ \bibnamefont {Chuang}}, \ and\ \bibinfo {author}
  {\bibfnamefont {R.}~\bibnamefont {Blatt}},\ }\href {\doibase
  10.1038/nature01336} {\bibfield  {journal} {\bibinfo  {journal} {Nature}\
  }\textbf {\bibinfo {volume} {421}},\ \bibinfo {pages} {48} (\bibinfo {year}
  {2003})}\BibitemShut {NoStop}%
\bibitem [{\citenamefont {Chiaverini}\ \emph {et~al.}(2004)\citenamefont
  {Chiaverini}, \citenamefont {Leibfried}, \citenamefont {Schaetz},
  \citenamefont {Barrett}, \citenamefont {Blakestad}, \citenamefont {Britton},
  \citenamefont {Itano}, \citenamefont {Jost}, \citenamefont {Knill},
  \citenamefont {Langer}, \citenamefont {Ozeri},\ and\ \citenamefont
  {Wineland}}]{Chiaverini-2004}%
  \BibitemOpen
  \bibfield  {author} {\bibinfo {author} {\bibfnamefont {J.}~\bibnamefont
  {Chiaverini}}, \bibinfo {author} {\bibfnamefont {D.}~\bibnamefont
  {Leibfried}}, \bibinfo {author} {\bibfnamefont {T.}~\bibnamefont {Schaetz}},
  \bibinfo {author} {\bibfnamefont {M.~D.}\ \bibnamefont {Barrett}}, \bibinfo
  {author} {\bibfnamefont {R.~B.}\ \bibnamefont {Blakestad}}, \bibinfo {author}
  {\bibfnamefont {J.}~\bibnamefont {Britton}}, \bibinfo {author} {\bibfnamefont
  {W.~M.}\ \bibnamefont {Itano}}, \bibinfo {author} {\bibfnamefont {J.~D.}\
  \bibnamefont {Jost}}, \bibinfo {author} {\bibfnamefont {E.}~\bibnamefont
  {Knill}}, \bibinfo {author} {\bibfnamefont {C.}~\bibnamefont {Langer}},
  \bibinfo {author} {\bibfnamefont {R.}~\bibnamefont {Ozeri}}, \ and\ \bibinfo
  {author} {\bibfnamefont {D.~J.}\ \bibnamefont {Wineland}},\ }\href
  {http://dx.doi.org/10.1038/nature03074} {\bibfield  {journal} {\bibinfo
  {journal} {Nature}\ }\textbf {\bibinfo {volume} {432}},\ \bibinfo {pages}
  {602} (\bibinfo {year} {2004})}\BibitemShut {NoStop}%
\bibitem [{\citenamefont {Friedenauer}\ \emph {et~al.}(2008)\citenamefont
  {Friedenauer}, \citenamefont {Schmitz}, \citenamefont {Glueckert},
  \citenamefont {Porras},\ and\ \citenamefont {Schaetz}}]{Friedenauer-2008}%
  \BibitemOpen
  \bibfield  {author} {\bibinfo {author} {\bibfnamefont {A.}~\bibnamefont
  {Friedenauer}}, \bibinfo {author} {\bibfnamefont {H.}~\bibnamefont
  {Schmitz}}, \bibinfo {author} {\bibfnamefont {J.~T.}\ \bibnamefont
  {Glueckert}}, \bibinfo {author} {\bibfnamefont {D.}~\bibnamefont {Porras}}, \
  and\ \bibinfo {author} {\bibfnamefont {T.}~\bibnamefont {Schaetz}},\ }\href
  {\doibase 10.1038/nphys1032} {\bibfield  {journal} {\bibinfo  {journal}
  {Nature Physics}\ } (\bibinfo {year} {2008}),\ 10.1038/nphys1032}\BibitemShut
  {NoStop}%
\bibitem [{\citenamefont {Martinis}(2009)}]{Martinis-review-2009}%
  \BibitemOpen
  \bibfield  {author} {\bibinfo {author} {\bibfnamefont {J.~M.}\ \bibnamefont
  {Martinis}},\ }\href {http://dx.doi.org/10.1007/s11128-009-0105-1} {\bibfield
   {journal} {\bibinfo  {journal} {Quantum Information Processing}\ }\textbf
  {\bibinfo {volume} {8}},\ \bibinfo {pages} {81} (\bibinfo {year}
  {2009})}\BibitemShut {NoStop}%
\bibitem [{\citenamefont {Kim}\ \emph {et~al.}(2010)\citenamefont {Kim},
  \citenamefont {Chang}, \citenamefont {Korenblit}, \citenamefont {Islam},
  \citenamefont {Edwards}, \citenamefont {Freericks}, \citenamefont {Lin},
  \citenamefont {Duan},\ and\ \citenamefont {Monroe}}]{Kim-2010}%
  \BibitemOpen
  \bibfield  {author} {\bibinfo {author} {\bibfnamefont {K.}~\bibnamefont
  {Kim}}, \bibinfo {author} {\bibfnamefont {M.-S.}\ \bibnamefont {Chang}},
  \bibinfo {author} {\bibfnamefont {S.}~\bibnamefont {Korenblit}}, \bibinfo
  {author} {\bibfnamefont {R.}~\bibnamefont {Islam}}, \bibinfo {author}
  {\bibfnamefont {E.~E.}\ \bibnamefont {Edwards}}, \bibinfo {author}
  {\bibfnamefont {J.~K.}\ \bibnamefont {Freericks}}, \bibinfo {author}
  {\bibfnamefont {G.-D.}\ \bibnamefont {Lin}}, \bibinfo {author} {\bibfnamefont
  {L.-M.}\ \bibnamefont {Duan}}, \ and\ \bibinfo {author} {\bibfnamefont
  {C.}~\bibnamefont {Monroe}},\ }\href {\doibase 10.1038/nature09071}
  {\bibfield  {journal} {\bibinfo  {journal} {Nature}\ }\textbf {\bibinfo
  {volume} {465}},\ \bibinfo {pages} {590} (\bibinfo {year}
  {2010})}\BibitemShut {NoStop}%
\bibitem [{\citenamefont {Gottesman}(1997)}]{gottesman-thesis}%
  \BibitemOpen
  \bibfield  {author} {\bibinfo {author} {\bibfnamefont {D.}~\bibnamefont
  {Gottesman}},\ }\emph {\bibinfo {title} {Stabilizer Codes and Quantum Error
  Correction}},\ \href {http://arxiv.org/abs/quant-ph/9705052} {Ph.D. thesis},\
  \bibinfo  {school} {Caltech} (\bibinfo {year} {1997})\BibitemShut {NoStop}%
\bibitem [{\citenamefont {Calderbank}\ \emph {et~al.}(1998)\citenamefont
  {Calderbank}, \citenamefont {Rains}, \citenamefont {Shor},\ and\
  \citenamefont {Sloane}}]{Calderbank-1997}%
  \BibitemOpen
  \bibfield  {author} {\bibinfo {author} {\bibfnamefont {A.~R.}\ \bibnamefont
  {Calderbank}}, \bibinfo {author} {\bibfnamefont {E.~M.}\ \bibnamefont
  {Rains}}, \bibinfo {author} {\bibfnamefont {P.~M.}\ \bibnamefont {Shor}}, \
  and\ \bibinfo {author} {\bibfnamefont {N.~J.~A.}\ \bibnamefont {Sloane}},\
  }\href {http://dx.doi.org/10.1109/18.681315} {\bibfield  {journal} {\bibinfo
  {journal} {IEEE Trans. Inf. Th.}\ }\textbf {\bibinfo {volume} {44}},\
  \bibinfo {pages} {1369} (\bibinfo {year} {1998})}\BibitemShut {NoStop}%
\bibitem [{\citenamefont {Grassl}(2007)}]{Grassl:codetables}%
  \BibitemOpen
  \bibfield  {author} {\bibinfo {author} {\bibfnamefont {M.}~\bibnamefont
  {Grassl}},\ }\href@noop {} {\enquote {\bibinfo {title} {Bounds on the minimum
  distance of linear codes and quantum codes},}\ }\bibinfo {howpublished}
  {Online available at \url{http://www.codetables.de}} (\bibinfo {year}
  {2007}),\ \bibinfo {note} {accessed on 2011-07-28}\BibitemShut {NoStop}%
\bibitem [{\citenamefont {Smolin}\ \emph {et~al.}(2007)\citenamefont {Smolin},
  \citenamefont {Smith},\ and\ \citenamefont {Wehner}}]{Smolin-2007}%
  \BibitemOpen
  \bibfield  {author} {\bibinfo {author} {\bibfnamefont {J.~A.}\ \bibnamefont
  {Smolin}}, \bibinfo {author} {\bibfnamefont {G.}~\bibnamefont {Smith}}, \
  and\ \bibinfo {author} {\bibfnamefont {S.}~\bibnamefont {Wehner}},\ }\href
  {http://link.aps.org/abstract/PRL/v99/e130505} {\bibfield  {journal}
  {\bibinfo  {journal} {Phys. Rev. Lett.}\ }\textbf {\bibinfo {volume} {99}},\
  \bibinfo {pages} {130505} (\bibinfo {year} {2007})}\BibitemShut {NoStop}%
\bibitem [{\citenamefont {Cross}\ \emph {et~al.}(2009)\citenamefont {Cross},
  \citenamefont {Smith}, \citenamefont {Smolin},\ and\ \citenamefont
  {Zeng}}]{Cross-CWS-2009}%
  \BibitemOpen
  \bibfield  {author} {\bibinfo {author} {\bibfnamefont {A.}~\bibnamefont
  {Cross}}, \bibinfo {author} {\bibfnamefont {G.}~\bibnamefont {Smith}},
  \bibinfo {author} {\bibfnamefont {J.~A.}\ \bibnamefont {Smolin}}, \ and\
  \bibinfo {author} {\bibfnamefont {B.}~\bibnamefont {Zeng}},\ }\href
  {http://dx.doi.org/10.1109/TIT.2008.2008136} {\bibfield  {journal} {\bibinfo
  {journal} {IEEE Trans. Inf. Th.}\ }\textbf {\bibinfo {volume} {55}},\
  \bibinfo {pages} {433} (\bibinfo {year} {2009})}\BibitemShut {NoStop}%
\bibitem [{\citenamefont {Chuang}\ \emph {et~al.}(2009)\citenamefont {Chuang},
  \citenamefont {Cross}, \citenamefont {Smith}, \citenamefont {Smolin},\ and\
  \citenamefont {Zeng}}]{Chuang-CWS-2009}%
  \BibitemOpen
  \bibfield  {author} {\bibinfo {author} {\bibfnamefont {I.~L.}\ \bibnamefont
  {Chuang}}, \bibinfo {author} {\bibfnamefont {A.~W.}\ \bibnamefont {Cross}},
  \bibinfo {author} {\bibfnamefont {G.}~\bibnamefont {Smith}}, \bibinfo
  {author} {\bibfnamefont {J.~A.}\ \bibnamefont {Smolin}}, \ and\ \bibinfo
  {author} {\bibfnamefont {B.}~\bibnamefont {Zeng}},\ }\href
  {http://dx.doi.org/10.1063/1.3086833} {\bibfield  {journal} {\bibinfo
  {journal} {J. Math. Phys.}\ }\textbf {\bibinfo {volume} {50}},\ \bibinfo
  {pages} {042109} (\bibinfo {year} {2009})}\BibitemShut {NoStop}%
\bibitem [{\citenamefont {Chen}\ \emph {et~al.}(2008)\citenamefont {Chen},
  \citenamefont {Zeng},\ and\ \citenamefont {Chuang}}]{Chen-Zeng-Chuang-2008}%
  \BibitemOpen
  \bibfield  {author} {\bibinfo {author} {\bibfnamefont {X.}~\bibnamefont
  {Chen}}, \bibinfo {author} {\bibfnamefont {B.}~\bibnamefont {Zeng}}, \ and\
  \bibinfo {author} {\bibfnamefont {I.~L.}\ \bibnamefont {Chuang}},\ }\href
  {http://link.aps.org/abstract/PRA/v78/e062315} {\bibfield  {journal}
  {\bibinfo  {journal} {Phys. Rev. A}\ }\textbf {\bibinfo {volume} {78}},\
  \bibinfo {pages} {062315} (\bibinfo {year} {2008})}\BibitemShut {NoStop}%
\bibitem [{\citenamefont {Yu}\ \emph {et~al.}(2008)\citenamefont {Yu},
  \citenamefont {Chen}, \citenamefont {Lai},\ and\ \citenamefont
  {Oh}}]{Yu-2008}%
  \BibitemOpen
  \bibfield  {author} {\bibinfo {author} {\bibfnamefont {S.}~\bibnamefont
  {Yu}}, \bibinfo {author} {\bibfnamefont {Q.}~\bibnamefont {Chen}}, \bibinfo
  {author} {\bibfnamefont {C.~H.}\ \bibnamefont {Lai}}, \ and\ \bibinfo
  {author} {\bibfnamefont {C.~H.}\ \bibnamefont {Oh}},\ }\href
  {http://link.aps.org/abstract/PRL/v101/e090501} {\bibfield  {journal}
  {\bibinfo  {journal} {Phys. Rev. Lett.}\ }\textbf {\bibinfo {volume} {101}},\
  \bibinfo {pages} {090501} (\bibinfo {year} {2008})}\BibitemShut {NoStop}%
\bibitem [{\citenamefont {Yu}\ \emph {et~al.}(2007)\citenamefont {Yu},
  \citenamefont {Chen},\ and\ \citenamefont {Oh}}]{Yu-Chen-Oh-2007}%
  \BibitemOpen
  \bibfield  {author} {\bibinfo {author} {\bibfnamefont {S.}~\bibnamefont
  {Yu}}, \bibinfo {author} {\bibfnamefont {Q.}~\bibnamefont {Chen}}, \ and\
  \bibinfo {author} {\bibfnamefont {C.~H.}\ \bibnamefont {Oh}},\ }\href
  {http://arxiv.org/abs/0709.1780} {\enquote {\bibinfo {title} {Graphical
  quantum error-correcting codes},}\ } (\bibinfo {year} {2007}),\ \bibinfo
  {note} {arXiv}\BibitemShut {NoStop}%
\bibitem [{\citenamefont {Grassl}\ and\ \citenamefont
  {Roetteler}(2008{\natexlab{a}})}]{Grassl-Roetteler-2008A}%
  \BibitemOpen
  \bibfield  {author} {\bibinfo {author} {\bibfnamefont {M.}~\bibnamefont
  {Grassl}}\ and\ \bibinfo {author} {\bibfnamefont {M.}~\bibnamefont
  {Roetteler}},\ }in\ \href {http://dx.doi.org/10.1109/ISIT.2008.4594996}
  {\emph {\bibinfo {booktitle} {Proceedings 2008 IEEE Int. Symp. Inf. Th. (ISIT
  2008)}}}\ (\bibinfo {address} {Toronto, Canada},\ \bibinfo {year} {2008})\
  pp.\ \bibinfo {pages} {300--304}\BibitemShut {NoStop}%
\bibitem [{\citenamefont {Grassl}\ and\ \citenamefont
  {Roetteler}(2008{\natexlab{b}})}]{Grassl-Roetteler-2008B}%
  \BibitemOpen
  \bibfield  {author} {\bibinfo {author} {\bibfnamefont {M.}~\bibnamefont
  {Grassl}}\ and\ \bibinfo {author} {\bibfnamefont {M.}~\bibnamefont
  {Roetteler}},\ }in\ \href {http://dx.doi.org/10.1109/ITW.2008.4578694} {\emph
  {\bibinfo {booktitle} {Proc. IEEE Inf. Th. Workshop 2008 (ITW 2008)}}}\
  (\bibinfo {address} {Porto, Portugal},\ \bibinfo {year} {2008})\ pp.\
  \bibinfo {pages} {396--400}\BibitemShut {NoStop}%
\bibitem [{\citenamefont {Grassl}\ \emph {et~al.}(2009)\citenamefont {Grassl},
  \citenamefont {Shor}, \citenamefont {Smith}, \citenamefont {Smolin},\ and\
  \citenamefont {Zeng}}]{Grassl-2009}%
  \BibitemOpen
  \bibfield  {author} {\bibinfo {author} {\bibfnamefont {M.}~\bibnamefont
  {Grassl}}, \bibinfo {author} {\bibfnamefont {P.}~\bibnamefont {Shor}},
  \bibinfo {author} {\bibfnamefont {G.}~\bibnamefont {Smith}}, \bibinfo
  {author} {\bibfnamefont {J.}~\bibnamefont {Smolin}}, \ and\ \bibinfo {author}
  {\bibfnamefont {B.}~\bibnamefont {Zeng}},\ }\href
  {http://link.aps.org/abstract/PRA/v79/e050306} {\bibfield  {journal}
  {\bibinfo  {journal} {Phys. Rev. A}\ }\textbf {\bibinfo {volume} {79}},\
  \bibinfo {pages} {050306} (\bibinfo {year} {2009})}\BibitemShut {NoStop}%
\bibitem [{\citenamefont {Li}\ \emph {et~al.}(2010{\natexlab{a}})\citenamefont
  {Li}, \citenamefont {Dumer},\ and\ \citenamefont
  {Pryadko}}]{Li-Dumer-Pryadko-2010}%
  \BibitemOpen
  \bibfield  {author} {\bibinfo {author} {\bibfnamefont {Y.}~\bibnamefont
  {Li}}, \bibinfo {author} {\bibfnamefont {I.}~\bibnamefont {Dumer}}, \ and\
  \bibinfo {author} {\bibfnamefont {L.~P.}\ \bibnamefont {Pryadko}},\ }\href
  {http://dx.doi.org/10.1103/PhysRevLett.104.190501} {\bibfield  {journal}
  {\bibinfo  {journal} {Phys. Rev. Lett.}\ }\textbf {\bibinfo {volume} {104}},\
  \bibinfo {pages} {190501} (\bibinfo {year} {2010}{\natexlab{a}})}\BibitemShut
  {NoStop}%
\bibitem [{\citenamefont {Li}\ \emph {et~al.}(2010{\natexlab{b}})\citenamefont
  {Li}, \citenamefont {Dumer}, \citenamefont {Grassl},\ and\ \citenamefont
  {Pryadko}}]{Li-Dumer-Grassl-Pryadko-2010}%
  \BibitemOpen
  \bibfield  {author} {\bibinfo {author} {\bibfnamefont {Y.}~\bibnamefont
  {Li}}, \bibinfo {author} {\bibfnamefont {I.}~\bibnamefont {Dumer}}, \bibinfo
  {author} {\bibfnamefont {M.}~\bibnamefont {Grassl}}, \ and\ \bibinfo {author}
  {\bibfnamefont {L.~P.}\ \bibnamefont {Pryadko}},\ }\href {\doibase
  10.1103/PhysRevA.81.052337} {\bibfield  {journal} {\bibinfo  {journal} {Phys.
  Rev. A}\ }\textbf {\bibinfo {volume} {81}},\ \bibinfo {pages} {052337}
  (\bibinfo {year} {2010}{\natexlab{b}})}\BibitemShut {NoStop}%
\bibitem [{\citenamefont {Kitaev}(2003)}]{kitaev-anyons}%
  \BibitemOpen
  \bibfield  {author} {\bibinfo {author} {\bibfnamefont {A.~Y.}\ \bibnamefont
  {Kitaev}},\ }\href {http://arxiv.org/abs/quant-ph/9707021} {\bibfield
  {journal} {\bibinfo  {journal} {Ann. Phys.}\ }\textbf {\bibinfo {volume}
  {303}},\ \bibinfo {pages} {2} (\bibinfo {year} {2003})}\BibitemShut {NoStop}%
\bibitem [{\citenamefont {Looi}\ \emph {et~al.}(2008)\citenamefont {Looi},
  \citenamefont {Yu}, \citenamefont {Gheorghiu},\ and\ \citenamefont
  {Griffiths}}]{Looi-2008}%
  \BibitemOpen
  \bibfield  {author} {\bibinfo {author} {\bibfnamefont {S.~Y.}\ \bibnamefont
  {Looi}}, \bibinfo {author} {\bibfnamefont {L.}~\bibnamefont {Yu}}, \bibinfo
  {author} {\bibfnamefont {V.}~\bibnamefont {Gheorghiu}}, \ and\ \bibinfo
  {author} {\bibfnamefont {R.~B.}\ \bibnamefont {Griffiths}},\ }\href
  {http://link.aps.org/abstract/PRA/v78/e042303} {\bibfield  {journal}
  {\bibinfo  {journal} {Phys. Rev. A}\ }\textbf {\bibinfo {volume} {78}},\
  \bibinfo {pages} {042303} (\bibinfo {year} {2008})}\BibitemShut {NoStop}%
\bibitem [{\citenamefont {Nielsen}\ and\ \citenamefont
  {Chuang}(2000)}]{Nielsen-book}%
  \BibitemOpen
  \bibfield  {author} {\bibinfo {author} {\bibfnamefont {M.~A.}\ \bibnamefont
  {Nielsen}}\ and\ \bibinfo {author} {\bibfnamefont {I.~L.}\ \bibnamefont
  {Chuang}},\ }\href@noop {} {\emph {\bibinfo {title} {Quantum Computation and
  Quantum Infomation}}}\ (\bibinfo  {publisher} {Cambridge Unive. Press},\
  \bibinfo {address} {Cambridge, MA},\ \bibinfo {year} {2000})\BibitemShut
  {NoStop}%
\bibitem [{\citenamefont {Wilde}(2009)}]{Wilde-2009}%
  \BibitemOpen
  \bibfield  {author} {\bibinfo {author} {\bibfnamefont {M.~M.}\ \bibnamefont
  {Wilde}},\ }\href {http://dx.doi.org/10.1103/PhysRevA.79.062322} {\bibfield
  {journal} {\bibinfo  {journal} {Phys. Rev. A}\ }\textbf {\bibinfo {volume}
  {79}},\ \bibinfo {pages} {062322} (\bibinfo {year} {2009})}\BibitemShut
  {NoStop}%
\bibitem [{\citenamefont {Calderbank}\ and\ \citenamefont
  {Shor}(1996)}]{Calderbank-Shor-1996}%
  \BibitemOpen
  \bibfield  {author} {\bibinfo {author} {\bibfnamefont {A.~R.}\ \bibnamefont
  {Calderbank}}\ and\ \bibinfo {author} {\bibfnamefont {P.~W.}\ \bibnamefont
  {Shor}},\ }\href {\doibase 10.1103/PhysRevA.54.1098} {\bibfield  {journal}
  {\bibinfo  {journal} {Phys. Rev. A}\ }\textbf {\bibinfo {volume} {54}},\
  \bibinfo {pages} {1098} (\bibinfo {year} {1996})}\BibitemShut {NoStop}%
\bibitem [{\citenamefont {Steane}(1996)}]{Steane-1996}%
  \BibitemOpen
  \bibfield  {author} {\bibinfo {author} {\bibfnamefont {A.~M.}\ \bibnamefont
  {Steane}},\ }\href {http://dx.doi.org/10.1103/PhysRevA.54.4741} {\bibfield
  {journal} {\bibinfo  {journal} {Phys. Rev. A}\ }\textbf {\bibinfo {volume}
  {54}},\ \bibinfo {pages} {4741} (\bibinfo {year} {1996})}\BibitemShut
  {NoStop}%
\bibitem [{\citenamefont {Grassl}\ and\ \citenamefont
  {Beth}(1997)}]{Grassl-1997}%
  \BibitemOpen
  \bibfield  {author} {\bibinfo {author} {\bibfnamefont {M.}~\bibnamefont
  {Grassl}}\ and\ \bibinfo {author} {\bibfnamefont {T.}~\bibnamefont {Beth}},\
  }\href {http://arxiv.org/abs/quant-ph/9703016} {\enquote {\bibinfo {title} {A
  note on non-additive quantum codes},}\ } (\bibinfo {year} {1997}),\ \bibinfo
  {note} {unpublished, arXiv:quant-ph/9703016}\BibitemShut {NoStop}%
\bibitem [{\citenamefont {Grassl}\ \emph {et~al.}(2002)\citenamefont {Grassl},
  \citenamefont {Klappenecker},\ and\ \citenamefont
  {R{\"o}tteler}}]{Grassl-Klappenecker-Roetteler-2002}%
  \BibitemOpen
  \bibfield  {author} {\bibinfo {author} {\bibfnamefont {M.}~\bibnamefont
  {Grassl}}, \bibinfo {author} {\bibfnamefont {A.}~\bibnamefont
  {Klappenecker}}, \ and\ \bibinfo {author} {\bibfnamefont {M.}~\bibnamefont
  {R{\"o}tteler}},\ }in\ \href@noop {} {\emph {\bibinfo {booktitle}
  {Proceedings of the 2002 IEEE International Symposium on Information
  Theory}}}\ (\bibinfo {address} {Lausanne, Switzerland},\ \bibinfo {year}
  {2002})\ p.~\bibinfo {pages} {45},\ \bibinfo {note} {preprint
  quant-ph/0703112}\BibitemShut {NoStop}%
\bibitem [{\citenamefont {Schlingemann}(2002)}]{Schlingemann-2002}%
  \BibitemOpen
  \bibfield  {author} {\bibinfo {author} {\bibfnamefont {D.}~\bibnamefont
  {Schlingemann}},\ }\href@noop {} {\bibfield  {journal} {\bibinfo  {journal}
  {Quantum Information and Computation}\ }\textbf {\bibinfo {volume} {2}},\
  \bibinfo {pages} {307} (\bibinfo {year} {2002})},\ \bibinfo {note} {preprint
  quant-ph/0111080}\BibitemShut {NoStop}%
\bibitem [{\citenamefont {Van~den Nest}\ \emph {et~al.}(2004)\citenamefont
  {Van~den Nest}, \citenamefont {Dehaene},\ and\ \citenamefont
  {De~Moor}}]{VandenNest-2004}%
  \BibitemOpen
  \bibfield  {author} {\bibinfo {author} {\bibfnamefont {M.}~\bibnamefont
  {Van~den Nest}}, \bibinfo {author} {\bibfnamefont {J.}~\bibnamefont
  {Dehaene}}, \ and\ \bibinfo {author} {\bibfnamefont {B.}~\bibnamefont
  {De~Moor}},\ }\href {http://dx.doi.org/10.1103/PhysRevA.69.022316} {\bibfield
   {journal} {\bibinfo  {journal} {Phys. Rev. A}\ }\textbf {\bibinfo {volume}
  {69}},\ \bibinfo {pages} {022316} (\bibinfo {year} {2004})}\BibitemShut
  {NoStop}%
\bibitem [{\citenamefont {Hein}\ \emph {et~al.}(2005)\citenamefont {Hein},
  \citenamefont {D\"ur}, \citenamefont {Eisert}, \citenamefont {Raussendorf},
  \citenamefont {Van~den Nest},\ and\ \citenamefont {Briegel}}]{Hein-2006}%
  \BibitemOpen
  \bibfield  {author} {\bibinfo {author} {\bibfnamefont {M.}~\bibnamefont
  {Hein}}, \bibinfo {author} {\bibfnamefont {W.}~\bibnamefont {D\"ur}},
  \bibinfo {author} {\bibfnamefont {J.}~\bibnamefont {Eisert}}, \bibinfo
  {author} {\bibfnamefont {R.}~\bibnamefont {Raussendorf}}, \bibinfo {author}
  {\bibfnamefont {M.}~\bibnamefont {Van~den Nest}}, \ and\ \bibinfo {author}
  {\bibfnamefont {H.-J.}\ \bibnamefont {Briegel}},\ }in\ \href
  {http://arxiv.org/abs/quant-ph/0602096} {\emph {\bibinfo {booktitle} {Quant.
  Comp., Algorithms. and Chaos: Proc. Int. School Physics "Enrico Fermi"}}},\
  Vol.\ \bibinfo {volume} {162}\ (\bibinfo  {publisher} {IOS Press},\ \bibinfo
  {address} {Amsterdam},\ \bibinfo {year} {2005})\ pp.\ \bibinfo {pages}
  {115--218}\BibitemShut {NoStop}%
\bibitem [{\citenamefont {Bouchet}(1993)}]{Bouchet-1993}%
  \BibitemOpen
  \bibfield  {author} {\bibinfo {author} {\bibfnamefont {A.}~\bibnamefont
  {Bouchet}},\ }\href {http://dx.doi.org/10.1016/0012-365X(93)90357-Y}
  {\bibfield  {journal} {\bibinfo  {journal} {Discrete Math.}\ }\textbf
  {\bibinfo {volume} {114}},\ \bibinfo {pages} {75 } (\bibinfo {year}
  {1993})}\BibitemShut {NoStop}%
\bibitem [{\citenamefont {Calderbank}\ \emph {et~al.}(1997)\citenamefont
  {Calderbank}, \citenamefont {Rains}, \citenamefont {Shor},\ and\
  \citenamefont {Sloane}}]{Calderbank-Rains-Shor-Sloane-1997}%
  \BibitemOpen
  \bibfield  {author} {\bibinfo {author} {\bibfnamefont {A.~R.}\ \bibnamefont
  {Calderbank}}, \bibinfo {author} {\bibfnamefont {E.~M.}\ \bibnamefont
  {Rains}}, \bibinfo {author} {\bibfnamefont {P.~W.}\ \bibnamefont {Shor}}, \
  and\ \bibinfo {author} {\bibfnamefont {N.~J.~A.}\ \bibnamefont {Sloane}},\
  }\href {http://dx.doi.org/10.1103/PhysRevLett.78.405} {\bibfield  {journal}
  {\bibinfo  {journal} {Phys. Rev. Lett.}\ }\textbf {\bibinfo {volume} {78}},\
  \bibinfo {pages} {405} (\bibinfo {year} {1997})}\BibitemShut {NoStop}%
\bibitem [{\citenamefont {Laflamme}\ \emph {et~al.}(1996)\citenamefont
  {Laflamme}, \citenamefont {Miquel}, \citenamefont {Paz},\ and\ \citenamefont
  {Zurek}}]{Laflamme-1996}%
  \BibitemOpen
  \bibfield  {author} {\bibinfo {author} {\bibfnamefont {R.}~\bibnamefont
  {Laflamme}}, \bibinfo {author} {\bibfnamefont {C.}~\bibnamefont {Miquel}},
  \bibinfo {author} {\bibfnamefont {J.~P.}\ \bibnamefont {Paz}}, \ and\
  \bibinfo {author} {\bibfnamefont {W.~H.}\ \bibnamefont {Zurek}},\ }\href
  {\doibase 10.1103/PhysRevLett.77.198} {\bibfield  {journal} {\bibinfo
  {journal} {Phys. Rev. Lett.}\ }\textbf {\bibinfo {volume} {77}},\ \bibinfo
  {pages} {198} (\bibinfo {year} {1996})}\BibitemShut {NoStop}%
\bibitem [{\citenamefont {Feng}\ and\ \citenamefont
  {Ma}(2004)}]{Feng:dec.2004}%
  \BibitemOpen
  \bibfield  {author} {\bibinfo {author} {\bibfnamefont {K.}~\bibnamefont
  {Feng}}\ and\ \bibinfo {author} {\bibfnamefont {Z.}~\bibnamefont {Ma}},\
  }\href {\doibase 10.1109/TIT.2004.838088} {\bibfield  {journal} {\bibinfo
  {journal} {Information Theory, IEEE Transactions on}\ }\textbf {\bibinfo
  {volume} {50}},\ \bibinfo {pages} {3323 } (\bibinfo {year}
  {2004})}\BibitemShut {NoStop}%
\bibitem [{\citenamefont {MacWilliams}\ and\ \citenamefont
  {Sloane}(1981)}]{MS-book}%
  \BibitemOpen
  \bibfield  {author} {\bibinfo {author} {\bibfnamefont {F.~J.}\ \bibnamefont
  {MacWilliams}}\ and\ \bibinfo {author} {\bibfnamefont {N.~J.~A.}\
  \bibnamefont {Sloane}},\ }\href@noop {} {\emph {\bibinfo {title} {The Theory
  of Error-Correcting Codes}}}\ (\bibinfo  {publisher} {North-Holland},\
  \bibinfo {address} {Amsterdam},\ \bibinfo {year} {1981})\BibitemShut
  {NoStop}%
\bibitem [{\citenamefont {Rains}\ and\ \citenamefont
  {Sloane}(1998)}]{Rains-Sloane-self-dual-codes-1998}%
  \BibitemOpen
  \bibfield  {author} {\bibinfo {author} {\bibfnamefont {E.~M.}\ \bibnamefont
  {Rains}}\ and\ \bibinfo {author} {\bibfnamefont {N.~J.~A.}\ \bibnamefont
  {Sloane}},\ }in\ \href {http://www2.research.att.com/~njas/doc/self.pdf}
  {\emph {\bibinfo {booktitle} {Handbook of coding theory}}},\ Vol.~\bibinfo
  {volume} {I},\ \bibinfo {editor} {edited by\ \bibinfo {editor} {\bibfnamefont
  {V.~S.~P.}\ \bibnamefont {and. W.~C.~Huffman}}}\ (\bibinfo  {publisher}
  {North-Holland},\ \bibinfo {address} {Amsterdam},\ \bibinfo {year} {1998})\
  pp.\ \bibinfo {pages} {177--294}\BibitemShut {NoStop}%
\bibitem [{Note1()}]{Note1}%
  \BibitemOpen
  \bibinfo {note} {In the literature such polynomials have also been called
  ``symmetric''. We prefer to reserve this term for the polynomials~(\ref
  {eq:symmetric-polynomial}) which correspond to symmetric circulant matrices.
  Palindromic polynomials have reflection symmetry with respect to their
  ``centers'', while Eq.~(\ref {eq:symmetric-polynomial}) corresponds to a
  symmetry with respect to the free term, with an implicit circulant
  symmetry.}\BibitemShut {Stop}%
\bibitem [{Note2()}]{Note2}%
  \BibitemOpen
  \bibinfo {note} {For polynomials $q(x)$ with multiple factors, distance
  estimates of quantum codes lead to the estimates of weight spectra of
  classical cyclic codes which contain the code generated by $q(x)$, which is
  beyond the scope of this work.}\BibitemShut {Stop}%
\bibitem [{\citenamefont {Ding}\ \emph {et~al.}(2002)\citenamefont {Ding},
  \citenamefont {Helleseth}, \citenamefont {Niederreiter},\ and\ \citenamefont
  {Xing}}]{Ding:OCT2002}%
  \BibitemOpen
  \bibfield  {author} {\bibinfo {author} {\bibfnamefont {C.}~\bibnamefont
  {Ding}}, \bibinfo {author} {\bibfnamefont {T.}~\bibnamefont {Helleseth}},
  \bibinfo {author} {\bibfnamefont {H.}~\bibnamefont {Niederreiter}}, \ and\
  \bibinfo {author} {\bibfnamefont {C.}~\bibnamefont {Xing}},\ }\href@noop {}
  {\bibfield  {journal} {\bibinfo  {journal} {IEEE Trans. Inf. Th.}\ }\textbf
  {\bibinfo {volume} {48}},\ \bibinfo {pages} {2679} (\bibinfo {year}
  {2002})}\BibitemShut {NoStop}%
\bibitem [{\citenamefont {Bravyi}\ and\ \citenamefont
  {Kitaev}()}]{Bravyi-1998}%
  \BibitemOpen
  \bibfield  {author} {\bibinfo {author} {\bibfnamefont {S.~B.}\ \bibnamefont
  {Bravyi}}\ and\ \bibinfo {author} {\bibfnamefont {A.~Y.}\ \bibnamefont
  {Kitaev}},\ }\href {http://arxiv.org/abs/quant-ph/9811052} {\enquote
  {\bibinfo {title} {Quantum codes on a lattice with boundary},}\ }\bibinfo
  {note} {Unpublished, arXiv:quant-ph/9811052}\BibitemShut {NoStop}%
\bibitem [{\citenamefont {Freedman}\ and\ \citenamefont
  {Meyer}()}]{Freedman-1998}%
  \BibitemOpen
  \bibfield  {author} {\bibinfo {author} {\bibfnamefont {M.~H.}\ \bibnamefont
  {Freedman}}\ and\ \bibinfo {author} {\bibfnamefont {D.~A.}\ \bibnamefont
  {Meyer}},\ }\href {http://arxiv.org/abs/quant-ph/9810055} {\enquote {\bibinfo
  {title} {Projective plane and planar quantum codes},}\ }\bibinfo {note}
  {Unpublished, arXiv:quant-ph/9810055}\BibitemShut {NoStop}%
\bibitem [{\citenamefont {Bombin}\ and\ \citenamefont
  {Martin-Delgado}(2007)}]{Bombin-2007}%
  \BibitemOpen
  \bibfield  {author} {\bibinfo {author} {\bibfnamefont {H.}~\bibnamefont
  {Bombin}}\ and\ \bibinfo {author} {\bibfnamefont {M.~A.}\ \bibnamefont
  {Martin-Delgado}},\ }\href {\doibase 10.1103/PhysRevA.76.012305} {\bibfield
  {journal} {\bibinfo  {journal} {Phys. Rev. A}\ }\textbf {\bibinfo {volume}
  {76}},\ \bibinfo {pages} {012305} (\bibinfo {year} {2007})}\BibitemShut
  {NoStop}%
\end{thebibliography}
%

\end{document}